%% file: mysemi-strong_LETZTE_KORREKTUR.tex
\documentclass[a4paper,10pt]{siamltex} 
\usepackage{amsmath,amssymb,amsfonts,graphicx,epsfig,color,url}
\usepackage{latexsym,mathrsfs,cite} 

\newcommand{\eps}{\varepsilon}

\definecolor{colorJRblue}{rgb}{0.,0.,1.}
\definecolor{colorJRred}{rgb}{1.,0.,0.}

\def\front{\mathrm{front}}
\def\base{\mathrm{base}}
\def\calO{\mathcal{O}}
\def\calL{\mathcal{L}}
\def\calN{\mathcal{N}}
\def\calM{\mathcal{M}}

\newcommand{\uw}{\check{u}}
\newcommand{\vw}{\check{v}}
\newcommand{\pw}{\check{p}}

\newcommand{\Uw}{\check{U}}
\newcommand{\Vw}{\check{V}}

\newcommand{\alw}{\check{\alpha}}
\newcommand{\bew}{\check{\beta}}

\newcommand{\us}{\hat{u}}
\newcommand{\ust}{u^\rmf}
\newcommand{\Vs}{\hat{V}}
\newcommand{\Us}{\hat{U}}
\newcommand{\vs}{\hat{v}}
\newcommand{\ps}{\hat{p}}
\newcommand{\qs}{\hat{q}}

\newcommand{\tU}{\tilde{U}}
\newcommand{\tV}{\tilde{V}}
\def\tF{\tilde{F}}

\newcommand{\R}{\mathbb{R}}
\newcommand{\rmd}{\mathrm{d}}

\newcommand{\rma}{\mathrm{a}}
\newcommand{\rmf}{\mathrm{f}}
\newcommand{\rms}{\mathrm{s}}

\newcommand{\sgn}{\mathrm{sgn}}
\newcommand{\calD}{\mathcal{D}}
\newcommand{\calF}{\mathcal{F}}
\newcommand{\pf}{\mathrm{pf}}
\newcommand{\sn}{\mathrm{sn}}

\def\H{J}
\def\const{e}

\newtheorem{Remark}{Remark}
\newtheorem{Proposition}{Proposition}

\newtheorem{Hypothesis}{Hypothesis}
\newtheorem{Definition}{Definition}

\begin{document}

\title{First and second order semi-strong interaction in reaction-diffusion systems}

\author{Jens D.M. Rademacher\thanks{Centrum Wiskunde \& Informatica, Science Park 123, 1098 XG Amsterdam, the Netherlands, rademach@cwi.nl 
}}

\maketitle

\begin{abstract}
Spatial scale separation often leads to sharp interfaces that can be fully localized pulses or transition layer fronts connecting different states. This paper concerns the asymptotic interaction laws of pulses and fronts in the so-called semi-strong regime of strongly differing diffusion lengths for reaction-diffusion systems in one space dimension. An asymptotic expansion and matching approach is applied in a model independent common framework. First order semi-strong interaction is introduced as a general interface interaction type. It is distinct from the semi-strong interaction studied over the past decade, which is referred to as `second order' here. Laws of motion are derived for pulses as well as fronts in abstract systems with attention to the effect of symmetries. 
First order interaction for pulses is shown to be gradient-like under conditions that are numerically checked for a class of equations including the Gray-Scott and Schnakenberg models.
\end{abstract}



\section{Introduction}
This paper concerns quasi-stationary patterns of reaction-diffusion systems in one space dimension with a scale dichotomy in the diffusion lengths of the form
\begin{equation}\label{e:rds}
\begin{array}{*{1}{rcl}}
\partial_t U &=& \hspace{3mm}D_u \partial_{xx}U + F(U,V;\eps)\\
\partial_t V &=& \eps^2 D_v \partial_{xx} V +G(U,V;\eps).
\end{array}
\end{equation}
Here $U\in\R^n$, $V\in\R^m$, and $D_u, D_v$ are diagonal matrices with positive entries, and $x\in D\subset\R$  an interval.

The parameter $0<\eps\ll 1$ is asymptotically small and yields the semi-strong limit as $\eps\searrow 0$. On the one hand, the localization of the $V$-components as $\eps\searrow 0$ to a point or jump discontinuity defines interface locations $x=r_j$, $j=1,\ldots,N$ of pulses or fronts, respectively. The shape of the interfaces is resolved on the small spatial scale $\xi = (x-r_j)/\eps$. On the other hand, the $U$-components remain continuous in $x$ as $\eps\searrow 0$ and smooth between interfaces. See Figure~\ref{f:profiles} for an illustration. 

\begin{figure}
   \centering
   \begin{tabular}{cc}
   \includegraphics[width=0.49\textwidth]{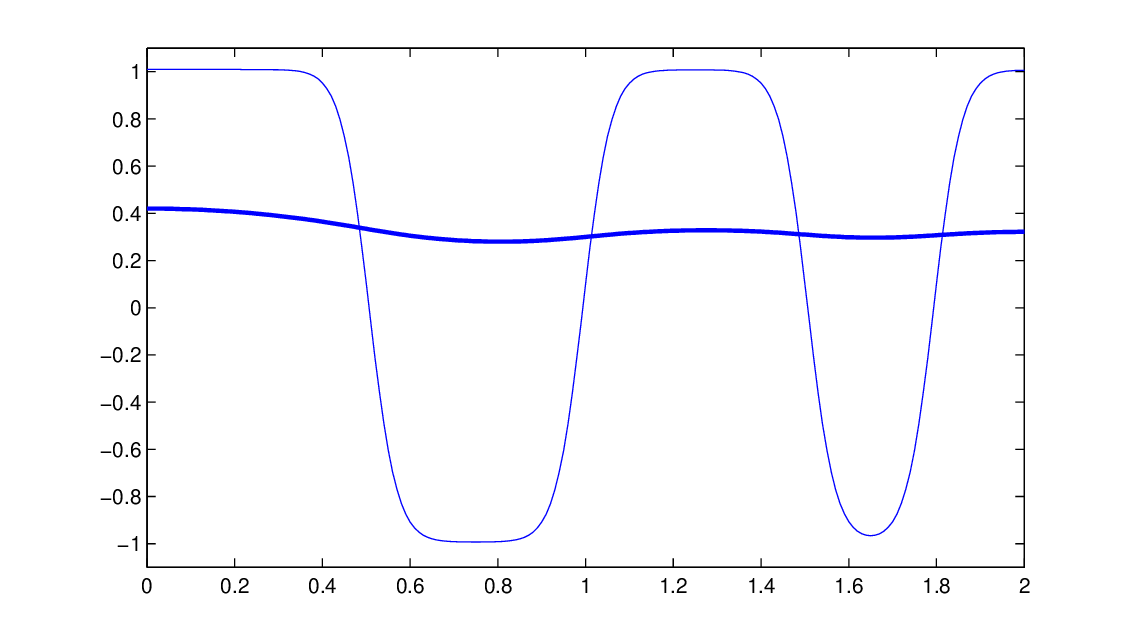}& \hspace{-8mm}
   \includegraphics[width=0.49\textwidth]{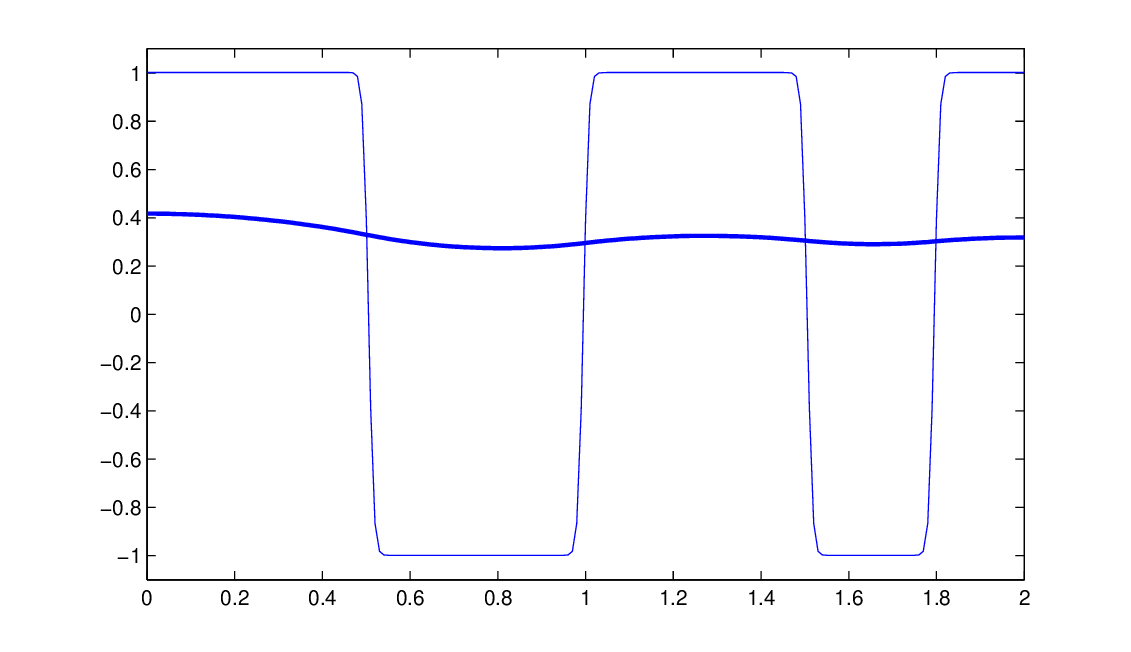}\\ 
   (a) & (b)\\
   \includegraphics[width=0.49\textwidth]{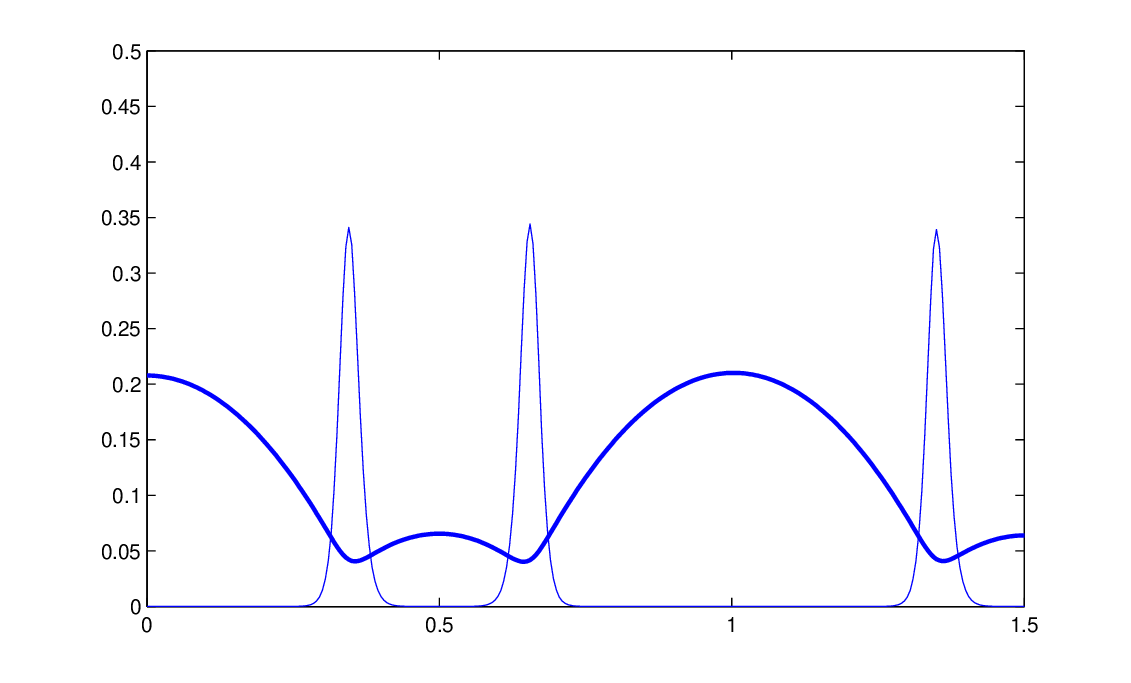} & \hspace{-8mm}
   \includegraphics[width=0.49\textwidth]{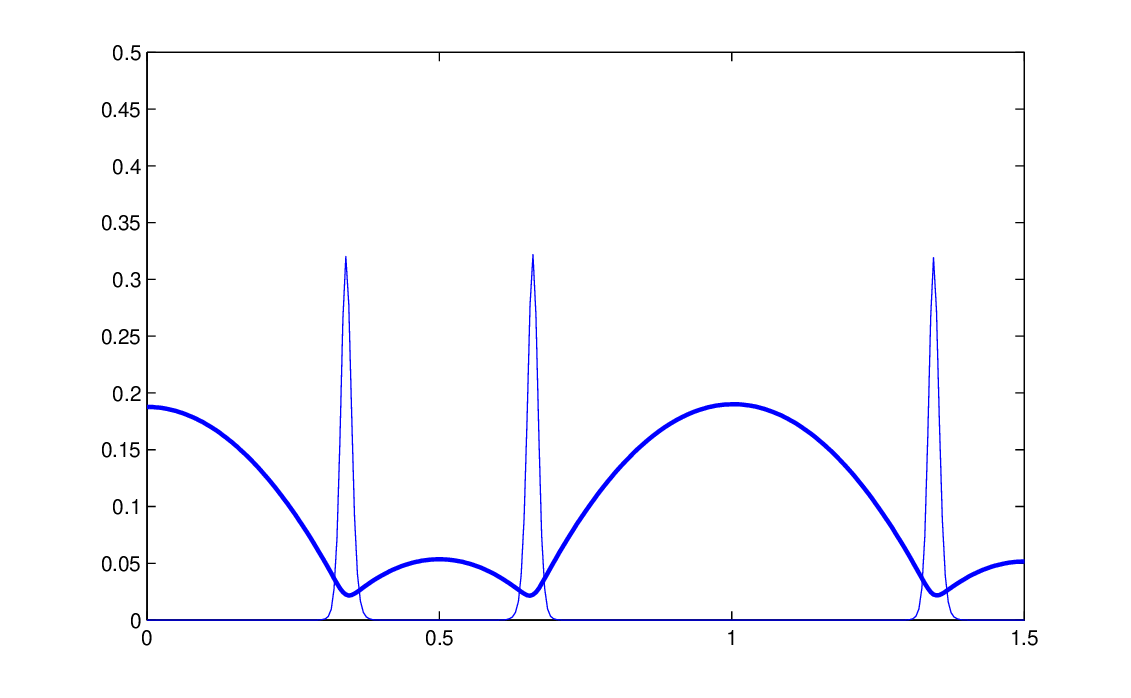} \\
   (c) & (d)
   \end{tabular}
   \caption{Snapshots of interface patterns for $D=[0,2]$ with Neumann boundary conditions. The thin lines show the $V$-component, the bold lines the $U$-component. (a,b): 4-front solutions to \eqref{e:ex-front}, where $\eps=0.05$ in (a) and $\eps=0.01$ in (b). 
 (c,d): 3-pulse solutions to \eqref{e:linsemi} with $V$ scaled to $\Vs= \eps V$, and $\eps=0.01$ in (c), $\eps=0.005$ in (d).}
   \label{f:profiles}
\end{figure}


\medskip
Systems of this form regularly occur in modelling, and slowly moving sharp interface patterns have been observed numerically as well as analytically treated for $n=m=1$ with a few exceptions. The literature is discussed in \S\ref{s:lit} below. 

\medskip
We are interested in the laws of motion of quasi-stationary interface patterns, which are characterized by $\frac{\rmd}{\rmd t} r_j \to 0$ as $\eps\to0$. Existence and smoothness of such solutions are \emph{assumed} (see Hypothesis~\ref{h:main}); proofs for some cases have been given in the literature. The purpose of this paper is threefold. First, to provide an abstract, model-independent unified framework for the study of semi-strong interaction of pulses and fronts.
Second, to introduce first order semi-strong interaction as a general type of interface interaction with  $r_j=r_j(\eps t)$ for \eqref{e:rds}. This kind of interaction is distinct from the more common `semi-strong interaction' \cite{DoeKap}, which has $r_j=r_j(\eps^2 t)$ and is referred to as `second order' here. 
Third, the distinct equations of motion for first and second order pulse and front interaction are derived and analyzed with attention to the effect of symmetry. 

It is shown that first order semi-strong pulse interaction is gradient-like (in the sense of a Lyapunov-functional) under conditions that are numerically checked for a class of equations including the Gray-Scott and Schnakenberg models.
To facilitate the analysis, necessary conditions on $F$ and $G$ for semi-strong interaction are derived. These `standard forms' distinguish the orders of interaction and the type of interfaces, and help to a priori determine the kind of interaction that can occur. 


\medskip
One may view semi-strong interaction as middle ground between weak interaction and strong interaction. \textit{Weak interaction} arises when all components localize at interfaces, which therefore interact only through tails that are exponentially close to homogeneous steady states. For instance, a 2-pulse in this case has a spatial profile that resembles a 2-homoclinic orbit, that is, it twice makes an excursion from a saddle equilibrium and passes close to the saddle inbetween. The resulting interaction law structure is universal and yields exponentially slow motion: in \eqref{e:rds} with $n=0$ roughly $r_j = r_j(\exp(-\kappa_j t/\eps))$, where $\kappa_j$ is essentially the slowest spatial convergence rate to the saddle equilibrium. Such solutions are also called `meta-stable'. See \cite{FuscoHale,CarrPego, Ei,sanTAMS, ZelikMielke} and also \cite{Promislow}. \textit{Strong interaction} occurs, when interface distances are on the smallest spatial scale, but
to the author's knowledge almost nothing is known in this case. An exception is that for scalar equations, such as Allen-Cahn, the gradient structure (and other ingredients) sometimes allows to study coarsening phenomena of domain walls as in \cite{Scheel}. `Pulse-splitting' and `pulse-annihilation' are strong interaction phenomena that have been observed in simulations in the semi-strong regime \cite{DoeKapPel, KoWaWe, KoWaWe1, KoWaWe2, MurOsi2, OsiSev, Rey}, but there are no rigorous results and essentially no theory. 

\medskip
In \textit{semi-strong interaction} the localization of only part of the components defines the interfaces. The components that do not localize drive the interaction, which is hence much stronger than the exponentially slow weak interaction. Indeed, while interface patterns in the Allen-Cahn equation interact weakly (are meta-stable),  it has been shown in \cite{HDKP} (see also \S \ref{s:bex-front} below) that a perturbation of the Allen-Cahn equation to the form \eqref{e:rds} increases the velocity to order $\eps^2$. In this paper we show that the semi-strong interaction laws are not universal, but generally come in two types: first and second order semi-strong interaction. Which of these occurs for a given model is not immediately clear, 
but the standard forms allow to a priori infer the scaling regimes in which either type of semi-strong interaction has a chance to occur. 
Moreover, the analysis explains order $\eps^2$ velocities in second order interaction as a result of additional symmetry in the limit $\eps\to 0$. And it shows a structural distinction: in the first order case each interface is to leading order driven only by its nearest neighbors, while in the second order case all other interfaces are relevant.

In various models the type of interaction is essentially determined by the amplitude scaling of a `feed' term: large feed generates first order interaction and small feed the slower second order interaction, which is in accordance with heuristic energy arguments. In particular, in these cases second order interaction is an asymptotic regime \emph{within} first order interaction.


\subsection{Relation to literature and models}\label{s:lit}


A dichotomy in diffusion lengths occurs naturally when $U$ and $V$ of \eqref{e:rds} model densities of particles with strongly differing mobilities. Examples in which semi-strong interaction occurs include models from chemistry such as the Brusselator, the Gray-Scott model, and the Schnakenberg model, as well as the phenomenological Gierer-Meinhardt model for sea-shell patterns, and a phenomenological gas discharge model. See \cite{DoeKap,HeDoeKap1,KoErWe}, respectively, and the references therein. For scalar $U$ and $V$, which is the predominant case in the literature, semi-strong interaction has been studied in \cite{DoeKap,DKP,KoWaWe1,KoWaWe2,WaSuRu, Rey}; an example with two-dimensional $U$ and scalar $V$ has been considered in \cite{DoeHeKap,HeDoeKap1, HDKP}. Single fronts have been considered in models with two-dimensional $V$ and scalar $U$ in \cite{Ike1,Ike2, Mill-ex, Mill-stab}. 
Semi-strong interaction is not restricted to reaction diffusion equations, but has also been studied in a nonlinear Schr\"odinger equation coupled to a temperature field \cite{MP08}.

\medskip
It is customary in modeling to set very small diffusion coefficients to zero; examples are the famous FitzHugh-Nagumo equations for action potentials in nerve axons (see \cite{FHN}) and the Oregonator model for the Belousov-Zhabotinsky reaction (see \cite{oreg}). This is justified in these cases, because the arising interfaces are not singular at $\eps=0$: they do not depend on both the large and the small scale of \eqref{e:rds}. In addition, the interfaces in these models move with nonzero speed $c$ in the limit $\eps=0$, and solve the ordinary differential equation in $x$ in the comoving space variable $x-ct$ when setting the time derivatives to zero.

\medskip
Essentially all previously studied interface motion in the semi-strong regime is second order in our notation. 
The exeption is \cite{Rey}, where first order interaction in the Gray-Scott model was considered. See also \cite{MurOsi2}. However, it has not been recognised as such and the connection and relation to second order semi-strong interaction has not been made -- the present paper fills this gap.

There is a fair amount of literature concerning stationary pulse and multi-pulse existence and stability in this context. See  \cite{DoeKap, DoeHeKap, HeDoeKap1,Ike1,Ike2, KoErWe,KoWaWe,KoWaWe1,KoWaWe2,Mill-ex,Mill-stab,MurOsi,OsiSev} and the references therein. Additional time scale separation is considered for instance in \cite{ChenWard, KoErWe, KoWaWe1}.
The existence of stationary pulse patterns in a semi-strong limit has been proven rigorously for a class of two-component systems in \cite{DoeKap} (and for special models in earlier work of these authors), and for a three-component model in \cite{DoeHeKap}. The arising equations of motion have been rigorously justified in a certain approximative sense and for special models in \cite{DKP, HDKP}, based on the renormalisation approach developed in \cite{Promislow}. 

\medskip
This paper is organized as follows. The next section introduces the concepts via basic examples. In \S\ref{s:main} the main results are summarized and an existence and smoothness hypothesis is formulated, which is used for the asymptotic expansion and matching arguments of subsequent sections. In \S\ref{s:standard} the standard forms are derived, and applied to a class of two-component models in \S\ref{s:apex}. The equations of motion for the second order semi-strong case are considered in \S\ref{s:gen-se-we} and for the first order semi-strong case in \S\ref{s:fast}. 

\medskip
\textbf{Acknowledgement.} This research has been supported in part by NWO cluster NDNS+. The author thanks Matthias Wolfrum, Julia Ehrt, Arjen Doelman and Michael Ward for helpful discussions. Special thanks go to  Peter van Heijster, Tasso Kaper and Michael Herrmann for comments on the manuscript.

\section{Basic examples}\label{s:bex}

To motivate and illustrate the subsequent more abstract analysis we consider the perhaps simplest models that support first and/or second order semi-strong interaction. In the following, subindices denote the order in the assumed $\eps$-expansions of the solutions; for instance, $V=V_0 + \eps V_1 + \calO(\eps^2)$.

\subsection{Second order front interaction}\label{s:bex-front}
Let us begin with the following two-component model that supports semi-strong front interaction.
\begin{equation}\label{e:ex-front}
\begin{array}{rl}
\partial_t U &= \partial_{xx}U - U + V\\
\partial_t V &= \eps^2 \partial_{xx}V  + V(1-V^2) + \eps U.
\end{array}
\end{equation}
This model is a perturbation of an Allen-Cahn equation and a case of the FitzHugh-Nagumo system with unusual parameters and scalings, and also a reduced version of the three-component system studied in \cite{HeDoeKap1}. In Figures~\ref{f:profiles}(a,b), \ref{f:4fronts} we plot examples of relevant solutions.

Assuming time-independence to leading order in $\eps$ implies fixed $V_0=V^*\in\{0,1,-1\}$ between interfaces and
\[
 \partial_{xx}U_0 = U_0 + V^*.
\]
This equation describes the (leading order) shape of the `background field' $U$ and has explicit solutions in terms of exponentials. We refer to these as large scale solutions.

\begin{figure}
   \centering
   \begin{tabular}{cc}
   \includegraphics[width=0.49\textwidth]{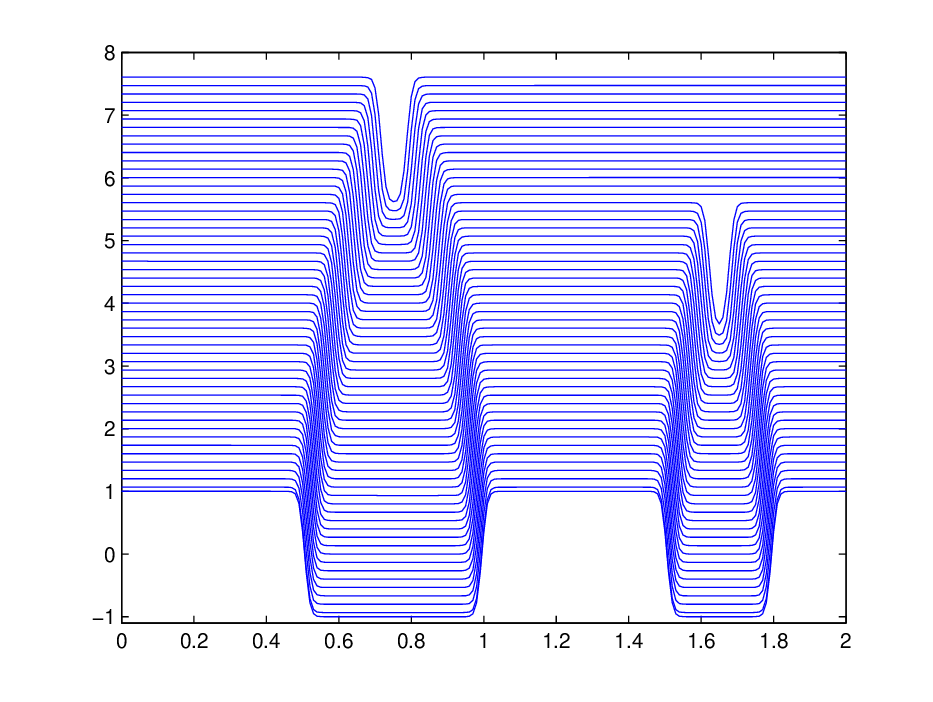} & \hspace{-8mm}
   \includegraphics[width=0.49\textwidth]{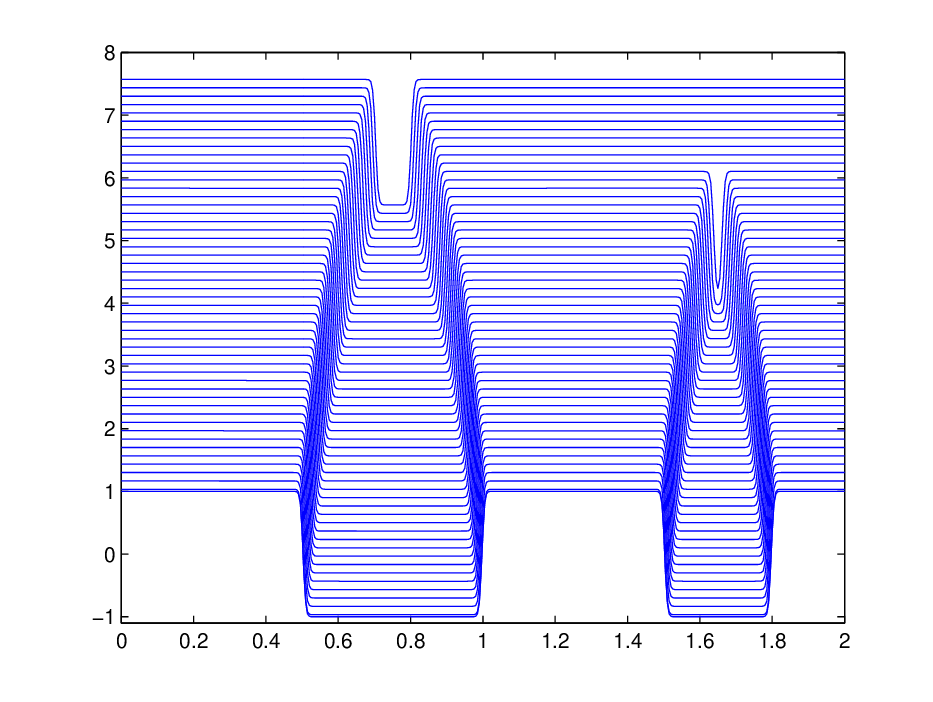} \\
   (a) & (b)
   \end{tabular}
   \caption{Staggered $V$-profiles of the evolving 4-front solution from Figure~\ref{f:profiles}(a,b) for (a) $\eps=0.01$ and up to time $2000$, (b) $\eps=0.005$ and up to time $8000$. The front locations move essentially on the same trajectories in both cases, which corroborates the $\eps^2$ time scale of second order semi-strong interaction. Before and after the coarsening event near $x=1.6$, where two fronts annihilate, the laws of motion from \S\ref{s:mainsnd} are expected to hold.}
   \label{f:4fronts}
\end{figure}

In order to resolve $V$ between the interfaces, it is appropriate to switch to the small spatial scale $\xi = x/\eps$. On this scale we use small letters $u, v$ and \eqref{e:ex-front} becomes
\begin{equation}\label{e:ex-front-sm}
\begin{array}{rl}
\eps^2\partial_t u &= \partial_{\xi\xi}u + \eps^2(-u + v)\\
\partial_t v &=\partial_{\xi\xi}v  + v(1-v^2) + \eps u.
\end{array}
\end{equation}

Assuming time-independence to leading order in $\eps$ implies constant $u_0$ (for boundedness), and
\begin{equation}\label{e:ex-front-sm0}
\partial_{\xi\xi}v_0  = -v_0(1-v_0^2),
\end{equation}
which has explicit heteroclinic solutions of $\tanh$-form that allow for connections between $V^*=\pm1$ in either direction. We refer to these as fronts, and they resolve the jumps at $\eps=0$ to leading order. Compare Figure~\ref{f:profiles}(a,b). 

\begin{Remark}\label{r:sym2nd}
We also infer that fronts connecting $\pm1$ must be stationary on the $\xi$-scale: a travelling wave ansatz $\xi \to \xi-ct$ with velocity $c$ introduces a friction term $-c\partial_\xi v$ in \eqref{e:ex-front-sm0}. Since \eqref{e:ex-front-sm0} is Hamiltonian and $v=\pm1$ have the same energy, heteroclinic connections for $c\neq 0$ cannot exist. Therefore, fronts move at most with velocity $\eps^2$ on the $x$-scale and thus system \eqref{e:ex-front} can only support second order front interaction. See Figure~\ref{f:4fronts}. 
\end{Remark}


\medskip
A sequence of $N$ fronts generates a leading order solution if the large scale solutions between these match appropriately at the front locations (and the boundary conditions if present). This yields a system of algebraic equations, analogous to that in \cite{HDKP}, whose solutions are (locally) parametrized by the $N$ front locations $r_j, j=1,\ldots,N$.  Accordingly, the linearization of \eqref{e:ex-front} in such a solution has, to leading order in $\eps$, a kernel of dimension $N$ spanned by the front translations. This in turn gives solvability conditions for a leading order construction of a time-dependent solution and thereby provides the laws of motion. Notably, here the term $\eps U$ in the $V$-equation comes into play. In \S\ref{s:gen-se-we} we provide details of this procedure for \eqref{e:rds} in standard form.

\subsection{First order pulse interaction}\label{s:bex-pulse}

Let us now turn to a simple example where first order interaction occurs, namely
\begin{equation}\label{e:linsemi}
\begin{array}{rl}
\partial_t U &= \partial_{xx}U + \alpha - V\\
\partial_t V &= \eps^2 \partial_{xx}V - V + UV^2.
\end{array}
\end{equation}
This is a simplified Schnakenberg model, cf.\ \S\ref{s:apex}, and closely related to the `linear' model in \cite{Rey}. See also \cite{OsiSev}.

\begin{figure}
   \centering
   \begin{tabular}{cc}
   \includegraphics[width=0.49\textwidth]{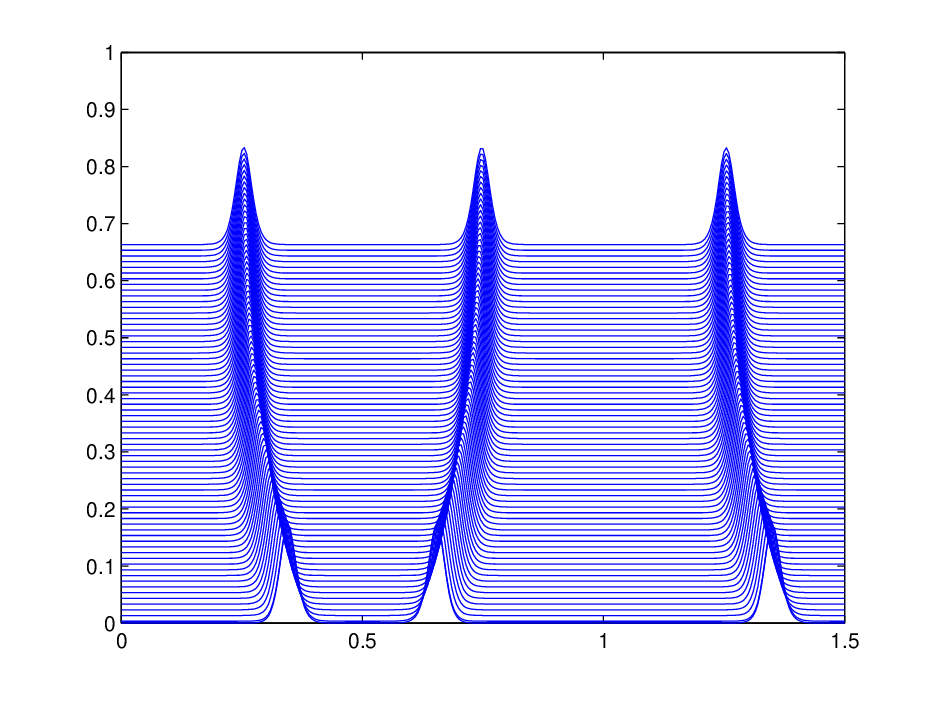} & \hspace{-8mm}\includegraphics[width=0.49\textwidth]{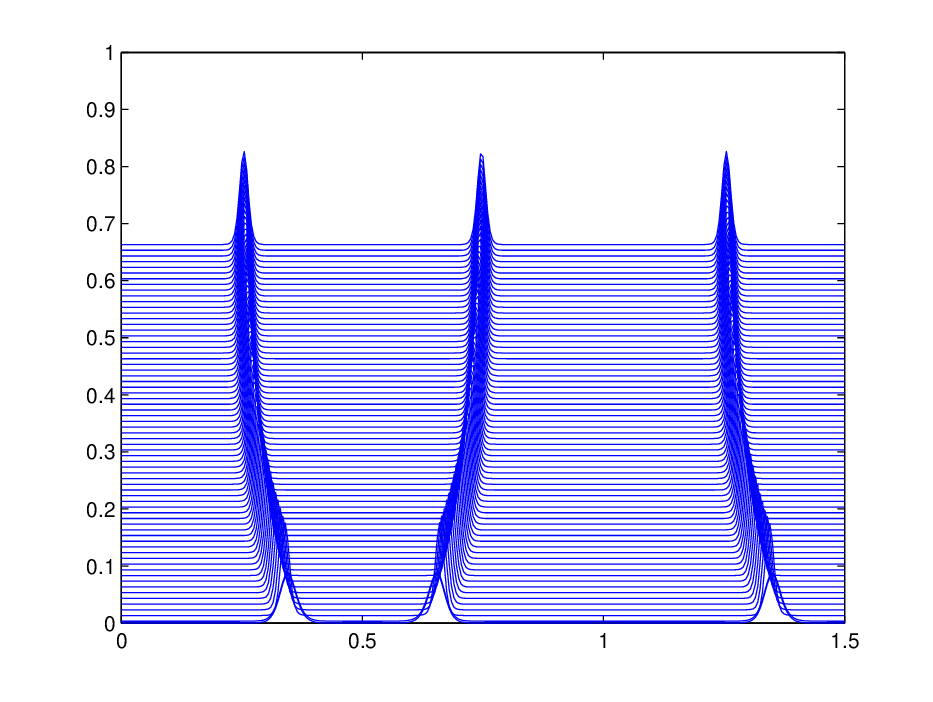} \\
   (a) & (b)
   \end{tabular}
   \caption{Staggered $\Vs$-profiles 
   of an evolving 3-pulse solution to \eqref{e:linsemi} for $\alpha=2.95$ and the initial condition plotted in Figure~\ref{f:profiles}(c,d). (a) $\eps=0.01$ over $200$ time units, (b) $\eps=0.005$ over $400$ time units. The pulse locations move essentially on the same trajectories in both cases, which corroborates the $\eps$ time scale of first order semi-strong interaction. Also the pulse amplitudes are similar, which reflects the scaling of $\Vs$.}
   \label{f:3pulses}
\end{figure}

As before, we assume time-independence to leading order in $\eps$. This gives $V_0\in\{0,1/U_0\}$ and 
\[
\partial_{xx} U_0 = -\alpha + V_0.
\]
Possible interfaces are again resolved on the small scale $\xi=x/\eps$, which, to leading order, gives constant $u_0$ (as in \eqref{e:ex-front-sm}$_1$) and 
\begin{equation}\label{e:linsemi-small1}
\partial_{\xi\xi} v_0 = v_0- u_0v_0^2.
\end{equation}
Since the latter only allows for homoclinic connections to $v_0=0$ (explicit $\cosh$-form), it follows that \eqref{e:linsemi} only supports pulse interaction. In contrast to the previous example this means $V_0=0$, which implies that the large scale solution  $U_0$ is independent of $V_0$ and therefore decoupled from the pulse locations. The result would be order $\eps^2$ motion  (see  Remark~\ref{r:sym2nd}) of pulses that is driven by a \emph{fixed} external field given by $U_0$, which satifies the boundary conditions on $D$. 
Let us have a closer look at how this degeneracy arises. Matching of small and large scale for derivatives means that $\partial_{\xi\xi}u = \eps^2(v-\alpha)$ should be written as a first order system in the form
\begin{equation}\label{e:fstslopes}
\begin{array}{rl}
\partial_\xi u &= \eps p\\
\partial_\xi p &= \eps(v-\alpha).
\end{array}
\end{equation}
Now, matching at a pulse location $x=r_*$ requires that the leading order large scale derivatives $\partial_x U_0(r_*\pm):=\lim_{\delta\searrow 0}\partial_x U_0(r_*\pm\delta)$ equal the limiting leading order small scale derivatives $\lim_{\xi\to\pm\infty}p_0(\xi)$. From the equation for $p$ in \eqref{e:fstslopes} we infer that the assumption $\lim_{\eps\to 0}\eps v =0$ causes the decoupling, as this leads to $\partial_x U_0(r_*+) = \partial_x U_0(r_*-)$. Indeed, bounded spikes carry zero measure in the limit and are invisible to $U_0$.

\begin{figure}
   \centering
   \includegraphics[width=0.49\textwidth]{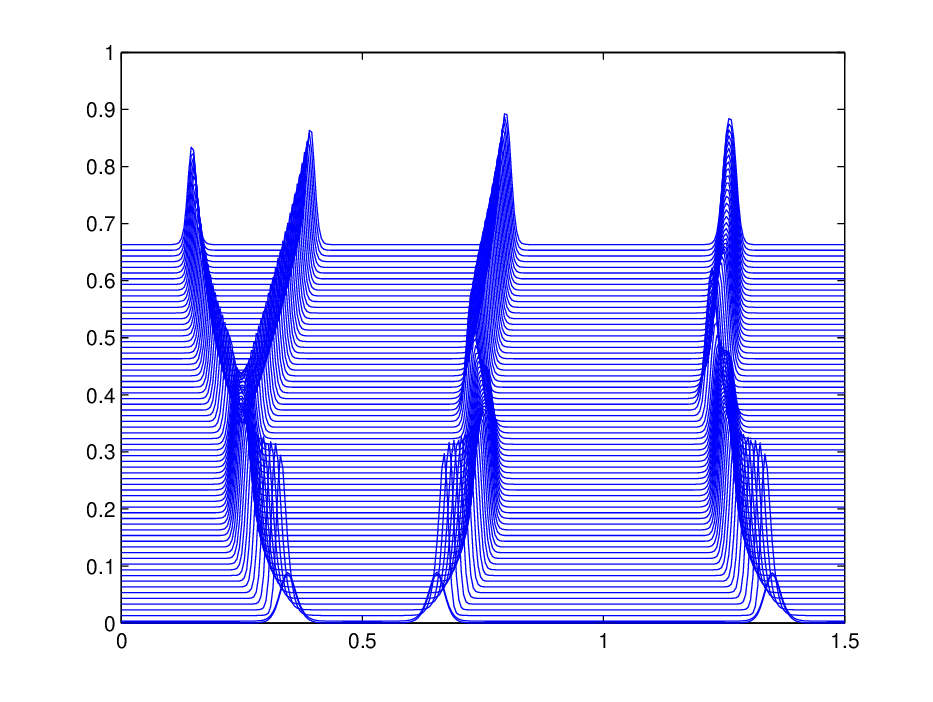}
   \caption{Staggered $\Vs$-profiles as in Figure~\ref{f:3pulses} for $\alpha=6$, $\eps=0.005$ and up to time $400$. At the pulse-splitting event near $x=0.3$, the asymptotic description breaks down.}
   \label{f:3pulse-split}
\end{figure}

\medskip
Therefore, let us make the ansatz $\Vs = \eps V$ in \eqref{e:linsemi}, which, for leading order stationary solutions, gives
\begin{equation}\label{e:linsemi2}
\begin{array}{rl}
\partial_{xx}U &= -\alpha + \eps^{-1}\Vs\\
\eps^2\partial_{xx}\Vs &= \Vs - \eps^{-1}U\Vs^2.
\end{array}
\end{equation}
Substituting expansions of $\Vs$ and $U$ in $\eps$ and comparing terms of equal order we find 
the solvability conditions $\Vs_0=0$ and $\Vs_1\in\{0, 1/U_0\}$. We shall argue below that $\Vs_1=0$ so that the leading order large scale problem still reads
\begin{equation}\label{e:linsemi-lg}
\partial_{xx}U_0 = -\alpha.
\end{equation}
However, on the small scale \eqref{e:linsemi-small1} turns into
\[
\partial_{\xi\xi} \vs = \vs - \eps^{-1} u \vs^2,
\]
so that the leading order term $u_0$ must vanish and we therefore set $u=\eps \us$ (but not $U=\eps\Us$). Hence, changing  \eqref{e:linsemi2} to the small scale gives to  leading order
\begin{equation}\label{e:linsemi-small2}
\begin{array}{rl}
\partial_{\xi\xi}\us_0 &= \vs_0\\
\partial_{\xi\xi}\vs_0 &= \vs_0 - \us_0\vs_0^2.
\end{array}
\end{equation}
Recall that matching of small and large scale at $x=r_*$ means $U_0(r_*)=u_0=0$ and
\[
P_\pm := \partial_x U(r_*\pm) = \lim_{\xi\pm\infty} \partial_{\xi}\us(\xi),
\] 
which are thus asymptotic boundary conditions for \eqref{e:linsemi-small2} that are completed by $\lim_{\xi\to\pm\infty}\vs_0(\xi)=0$ for pulse-solutions. While $V_1=1/U_0$ appeared as a possible solution above, it is ruled out by the constraint $U_0=0$ at pulse locations.

Let us consider \eqref{e:linsemi-small2} as a 4-dimensional first order ODE with $\ps_0=\partial_\xi \us_0$ and $\qs_0=\partial_\xi \vs_0$. It has the two-dimensional invariant space $\{\vs_0=\qs_0=0\}$ consisting of affine $\us_0(\xi) = \ps_0(0) \xi + \us_0(0)$. It is straightforward to compute that this space is normally hyperbolic and each point has one-dimensional stable and one-dimensional unstable manifolds. Pulse-solutions lie in the intersection of the resulting three-dimen\-sional (center-) stable and unstable manifolds, which in four dimensions is generically a two-dimensional transverse intersection. This means that we can expect at best a curve of $P_\pm$-values for which pulse-solutions exist. In fact, due to symmetry, this curve is $\{P_+ + P_- =0\}$. However, the slopes of the large scale solutions already exhaust the two-dimensional $P_\pm$-space, hence  \eqref{e:linsemi-small2} does not provide sufficiently many solutions. 

The problem is that \eqref{e:linsemi-small2} enforces motion of at least order $\eps$ on the $\xi$-scale, which is appropriate for second order semi-strong interaction. In the present case, however, we have to look for first order semi-strong interaction. This is done by allowing for a comoving frame $\xi \to \xi-ct$ before reducing to the leading order equations, which turns \eqref{e:linsemi-small2} into  
\begin{equation}\label{e:linsemi-small3}
\begin{array}{rl}
\partial_{\xi\xi}\us_0 &= \vs_0\\
\partial_{\xi\xi}\vs_0 &= -c\partial_\xi \vs_0 + \vs_0 - \us_0\vs_0^2.
\end{array}
\end{equation}
Indeed, the resulting motion is first order semi-strong as corroborated in Figure~\ref{f:3pulses}. In Figure~\ref{f:profiles}(c,d) we plot the solution profiles, which illustrate the large and small scale separation discussed here. In Figure~\ref{f:linear-bif} we plot a bifurcation diagram for \eqref{e:linsemi-small3} as explained in \S\ref{s:corenum}.  Figure~\ref{f:3pulse-split} shows the typical pulse-splitting phenomenon in equations such as \eqref{e:linsemi} for large `feed' $\alpha$, which cannot be resolved by the present analysis.

\subsubsection{Second order pulse interaction}
Numerical simulations show that for small $\alpha$ the phenomenology is more akin to the front interaction of the previous example and involves much slower pulse motion as well as coarsening by loss of pulses. This motivates to look for second order interaction for suitable $\alpha = o(1)$. Using the abstract `standard forms' derived in \S\ref{s:standard} below it is straightforward that the correct scalings for this are $\alpha=\eps^{1/2}\alw$, $U = \eps^{1/2}\Uw$, $V=\eps^{-1/2}\Vw$. See \S\ref{s:apex}.

Substituting this ansatz recovers \eqref{e:linsemi-lg} for $\Uw_0$, that is,
\[
\partial_{\xi\xi}\Uw_0 = -\alw,
\]
and also \eqref{e:linsemi-small3}$_2$ is recovered as
\[
\partial_{\xi\xi}\vw_0 = -c\partial_\xi \vw_0 + \vw_0 - \uw_0\vw_0^2.
\]
However, using \eqref{e:fstslopes}, we find that \eqref{e:linsemi-small3}$_1$ is transformed into
\begin{align*}
\partial_{\xi}\uw_0 &= 0\\
\partial_\xi\pw_0 &= -\vw_0.
\end{align*}
Hence, $\uw_0$ is constant and enters into the Hamiltonian equation for $\vw_0$ only as a parameter. Therefore, $c=0$ is required for homoclinic solutions, which implies \emph{second order} semi-strong interaction as in \S\ref{s:bex-front}. Recall Remark~\ref{r:sym2nd}. One may interpret this as slowing down of interfaces due to the small energy influx through $\eps^{1/2}\alw$, which is also in accordance with the aforementioned coarsening in this regime.

\section{Main results}\label{s:main}

In this section we summarize the main results for \eqref{e:rds}. We start by formulating the basic standing set of assumptions concerning existence and smoothness of quasi-stationary interface patterns. Its main purpose is to allow for asymptotic expansions and matching required later on. The part on the spectrum is used to derive the laws of motion in the second order case.

\begin{Hypothesis}\label{h:main}
There exist $T>0$, $V^\base, V^\front\in \R^m$, $\eps_0>0$, an open set of interface locations $r_{0,1}<\ldots < r_{0,N}$ in $\mathrm{int}(D)$, 
and a family of solutions $(U_\eps, V_\eps)$ to \eqref{e:rds} with the following properties for $\eps \in[0,\eps_0)$, $x\in D$ uniformly in $t\in[0,T]$. 

\noindent It is quasi-stationary, that is, $\partial_t(U_\eps,V_\eps) =\calO(\eps)$, bounded, 
twice continuously differentiable in $(\eps,x)\neq(0,r_{0,j})$, $j=1,\ldots, N$, and the limiting $U_0$ is continuous and non-constant in $D$ while $V_0=V^\base$ or $V_0=V^\front$ for $x\neq r_{j,0}$. 

\noindent For each $r_{0,j}$,  $j=1,\ldots,N$, with $x$ rescaled to $\eps\xi + r_{0,j}$ the family is twice continuously differentiable in $(\eps, \xi)$. In addition, the $L^2$-spectrum of the linearisation of \eqref{e:rds} in this rescaled family consists of $N$ eigenvalues of order $\eps$ and the remaining spectrum is stable and uniformly bounded away from the imaginary axis.
\end{Hypothesis}

%
This paper is not concerned with sufficient conditions under which Hypothesis~\ref{h:main} holds. See \cite{DKP, HDKP} for such results for certain models. Rather, we 
focus on the implications it has in allowing for asymptotic expansions. One may think of a center-manifold near the singular limit $\eps\searrow0$. 


Here and in the following we write $x\pm$ for the left and right limits $\lim_{\delta\searrow 0} x\pm\delta$.  Let $\bar{V}$ denote the first component of a vector $V$.

\begin{Definition}\label{d:semi} For solutions from Hypothesis~\ref{h:main}:
\begin{enumerate}
\item An interface is called a pulse if $V_0(r_{0,j}\pm)=V^\base$ and a front if\\ $|V_0(r_{0,j}+)-V_0(r_{0,j}-)|=|V^\base-V^\front|$.
\item For $\eps\in[0,\eps_0)$, $t\in[0,T]$ let $r_j(t)=r_{\eps,j}(t)$ be $C^1$ in $\eps$ (see remark below) such that for a pulse $\partial_x \bar{V}_\eps(r_j(t))=0$ and a front $|\bar{V}_\eps(r_j(t))|=\frac 1 2 |\bar{V}^\front-\bar{V}^\base|$. \label{i:loc}
\item With $r_j$ defined as in item \ref{i:loc}, interfaces are said to interact $k$-th order semi-strongly if there is non-trivial $R=(R_j)_{j=1}^N:\R^N\to \R^N$ such that\\ $\frac{\rmd}{\rmd t}r_j = \eps^k R_j(r_1,\ldots, r_N) + o(\eps^{k})$. 
\end{enumerate}
\end{Definition}

\medskip

\begin{Remark}
Concerning item \ref{i:loc} of Definition~\ref{d:semi}. Hypothesis~\ref{h:main} allows for smoothness in $\eps$, because the small scale solutions $v_\eps(\xi) = V_\eps(\eps\xi+r_{0,j})$ are smooth in $(\eps,\xi)$ at $\eps=0$. Hence, each local extremum of the interface profile $v_0$ perturbs for $0<\eps\ll 1$ to a unique curve of local extrema of $v_\eps$ that is smooth in $\eps$; analogously for the `mean value' of fronts.

Other choices for the definition of the interface locations for $\eps>0$ correspond to coordinate changes for the vector field $R$. 
\end{Remark}

We refer to $V^\base$ as the background state and, to ease notation, assume without loss of generality (by shifting $V$) that $V^\base=0$. 

\subsection{Standard forms} 

The discussion of the examples in \S\ref{s:bex} highlights that $F,G$ in \eqref{e:rds} for bounded quasi-stationary interface patterns naturally possess a singular expansion in $\eps$. Assuming integer exponents larger than or equal $-2$ and combining this with Taylor expansions in $U,V$, we derive in \S\ref{s:standard} that Hypothesis~\ref{h:main} yields
\begin{equation}\label{e:stand}
\begin{array}{*{1}{rcl}}
  \partial_t U &=& D_u \partial_{xx}U + H(U,V; \eps) + \eps^{-1}\left(F^\rms(U,V) + \eps^{-1}F^\rmf(U,V)U\right)V\\
  \partial_t V &=& \eps^2 D_v \partial_{xx}V + \eps E(U,V; \eps) + \left(G^\rms(U,V) + \eps^{-1}G^\rmf(U,V)U\right)V,
\end{array}
\end{equation}
where $E, H$ are smooth at $\eps=0$ and for fronts $F^\rmf, G^\rmf, F^\rms, G^\rms$ vanish at $V^\front$. The superscripts `s' and `f' indicate terms that are relevant for second and first order interaction, respectively, as discussed below. Note that \eqref{e:stand} is \emph{not} sufficient for existence of quasi-stationary patterns.

Additional conditions are (1) for $V^*\in\{0,V^\front\}$, if $G^\rmf(U,V^*)=G^\rms(U,V^*)= 0$ then also $E(U,V^*;0)=0$; (2)  $G^\rmf$ and $F^\rmf$ vanish at $V=0$ or else $V_1\equiv 0$. 

\begin{Remark}\label{r:sing}
Singular terms of order $\eps^{-j}$, $j>1$ on the right hand side of \eqref{e:stand}$_1$ are consistent only with first order semi-strong interaction, and these must possess factors $U^{j-1}$ and $V^j$ (or else $V$ must be order $\eps^j$). In  \eqref{e:stand}$_2$ the same holds when incrementing the orders of the factors by one.
Here we only consider $j\leq 2$ as this covers the concrete models of \S\ref{s:bex} and \S\ref{s:apex}.
\end{Remark}


\medskip
\emph{First order} interaction typically requires $G^\rmf\neq 0$. Specifically, symmetric pulses and $F^\rmf\neq 0$, $G^\rmf\equiv 0$ leads to trivial interaction. See \S\ref{s:first-small}. Note that the standard form covers the example \eqref{e:linsemi} with $F^\rmf=0$ and various other models as discussed in \S\ref{s:apex}. First order interaction for fronts has not been found in a specific equation to the authors knowledge, but arises naturally from the analogy to second order front interaction. 

\medskip
\emph{Second order} interaction for symmetric pulses requires $F^\rms\neq 0$, and otherwise nontrivial interaction is driven by $E$ as in the example \eqref{e:ex-front}. We refer to \eqref{e:stand} with $F^\rmf=G^\rmf=0$ as the standard form for second order semi-strong interaction. Indeed, all models of the type \eqref{e:rds}, where second order semi-strong interaction has been found, have this form. Compare \eqref{e:ex-front}, \eqref{e:linsemi2} and \S\ref{s:apex} below. In particular, setting $F^\rmf=G^\rmf=0$ covers the `normal form' for semi-strong pulse interaction in two-component models proposed in \cite{DoeKap}. This has $n=m=1$, $\eps\to \eps^2$ and
\[
H=-\eps\mu U, \; G^\rms= g(U)V-1, \; F^\rms = f(U)V, \; E\equiv 0.
\]
Under conditions on $f,g$ this form is in fact sufficient for the existence of pulses as in Hypothesis~\ref{h:main}. See \cite{DoeKap} and note that the Gierer-Meinhardt model in \cite{DKP,KoWaWe} also has second order standard form. So does the three component model studied in \cite{DoeHeKap,HeDoeKap1,HDKP}, which has fronts, and where $U=(w_1,w_2)\in\R^2$, $V\in\R$, $V^\front=2$ (in the reference $\tilde V=V-1$) and 
\[
F^\rms \equiv 0, \; E= -(\alpha w_1 + \beta w_2 + \gamma), \; G^\rms = V-1\,, \; H = ( (V-w_1)/\tau, (V-w_2)/\theta).
\]

\subsection{Laws of motion} 

We summarize the results on the laws of motion for interface interaction derived in \S\ref{s:standard} and \S\ref{s:gen-se-we}. 
We immediately note that, due to spatial translation symmetry, single interfaces on periodic $D$ or $D=\R$ have constant leading order velocity, and reflection symmetric pulses are stationary due to reflection symmetry of \eqref{e:rds} in $x$.

\medskip
On the `large' $x$-scale, for $x\in(r_j,r_{j+1})$, we write solutions $U_\eps, V_\eps$ from Hypothesis~\ref{h:main} as  $U_\eps = U_j=U_{0,j}+\eps U_{1,j} + \calO(\eps^2)$, $V_\eps = V_j= V_{0,j} + \eps V_{1,j} + \calO(\eps^2)$ for $j=0,\ldots, N+1$ with boundaries at $r_0, r_{N+1}$, if present. On the `small' spatial scale $\xi_j= (x-r_j)/\eps$ we use small letters $u,v$ and omit the index $j$ on $\xi$ in the following.


\subsubsection{Second order semi-strong interaction}\label{s:mainsnd}
In this case $F^\rmf=G^\rmf=0$ and the interface at $x=r_j$ is a heteroclinic or homoclinic orbit $v_{0,j}$ of
\begin{equation}\label{e:sndinter}
  D_v\partial_{\xi\xi}v_{0,j} + G^\rms(a_j,v_{0,j})v_{0,j} =0,
\end{equation}
where $a_j:=u_{0,j}$ is constant in $\xi$ and thus acts as a parameter. It turns out that the next order in $u$ is given by
\begin{equation}\label{e:sewe-sol-u}
u_{1,j}(\xi) = u_{1,j}(0) + \partial_\xi u_{1,j}(0)\xi - \int_0^{\xi} \int_0^\zeta D_u^{-1} F^\rms(a_j,v_{0,j}(\eta))v_{0,j}(\eta) \rmd\eta \rmd\zeta.
\end{equation}
\begin{Proposition}\label{p:sndlaw}
Assume Hypothesis~\ref{h:main} for \eqref{e:stand} with $G^\rmf=F^\rmf\equiv 0$. Then on the time scale $\tau = \eps^{-2} t$ it holds for each $j=1,\ldots,N$ that
\begin{equation}\label{e:slowmot1a}
\frac{\rmd}{\rmd\tau} r_{0,j} = - \langle \partial_u G^\rms(a_j,v_{0,j};0) [u_{1,j}, v_{0,j}] + E(a_j,v_{0,j};0), w_j \rangle / \langle \partial_\xi v_{0,j}, w_j \rangle,
\end{equation}
where $u_{1,j}$ is given by \eqref{e:sewe-sol-u} and $w_j$ spans the kernel of the $L^2$-adjoint of 
\[
\calL_j= D_v \partial_{\xi\xi} + G^\rms(a_j,v_{0,j})+ \partial_v G^\rms(a_j,v_{0,j})v_{0,j}.
\] 
\end{Proposition}
The proof is given in \S\ref{s:gen-se-we}.
Note that $\ker \calL_j$ is spanned by the translation mode $\partial_\xi v_{0,j}$.
In practice, if all interfaces are pulses then $v_{0,j} = v_{0,1}$, and if all are fronts, then $v_{0,2j+1} = v_{0,1}$, $v_{0,2j} = v_{0,2}$. 

At regular points of the matching problem between $x$- and $\xi$-scales, the right hand side of \eqref{e:slowmot1a} yields a vector field and generalizes the equations of motion reported in the literature. 
%
Together with matching the $V$-components, the leading order problem at one interface typically involves all others so that $R_j$ in Definition~\ref{d:semi} depends on all $r_k$ for $j,k=1,\ldots, N$.

\medskip
Let us briefly consider the time scale of motion. Similar to the example discussed in \S\ref{s:bex-pulse}, for $c= 0$ the $(v_0,q_0)$-equations are reversibly symmetric (by reflection symmetry $x\to -x$ of \eqref{e:rds}) \emph{also} when including the matching conditions. Hence, \emph{pulses} (homoclinic orbits) that are reflection symmetric about the $v_0$-subspace are generically robust (codimension zero), and persistent under perturbations of $u_0(0)=U(r_j)$. See \cite{Devaney}. Thus, $c=\calO(\eps)$ so that motion is order (at least) $\eps^2$ on the $x$-scale. See \S\ref{s:gen-se-we}. Moreover, for $m=1$ the equations are in addition Hamiltonian so that $c=\calO(\eps)$ for \emph{any} homoclinic solution. 

Concerning \emph{fronts}, for $m=1$ \eqref{e:sndinter} heteroclinic solutions whose asymptotic states have different energy (as in \cite{Ike1,Ike2,Mill-ex,Mill-stab}) cannot be stationary. However, two equilibria having the same energy is not structurally stable, so that locking motion at order $\eps^2$ for fronts requires additional structure in $G^\rms$, such as symmetry in $u$ and $v$. A trivial case is when $G^\rms(u,v)$ is independent of $u$ as in \S\ref{s:bex-front} and \cite{DoeHeKap,HeDoeKap1,HDKP}. 

\medskip
Generally, self-adjoint $\calL$ and symmetry of the interfaces allow to simplify the equations of motion and distinguish fronts and pulses. For instance, for even pulses the term in $E$ vanishes, while it essentially drives the motion of odd fronts in case of the aforementioned symmetry. See \S\ref{s:simpsym}. 
It is one of the strengths of the model independent approach that it provides a common framework, linking the results for different models and patterns from the literature. 

\subsubsection{First order semi-strong interaction} In this case the interface at $r_j$ and its leading order velocity $c_j$ are determined simultaneously by the boundary value problem
\begin{equation}\label{e:sest-gen}
\begin{array}{*{1}{rcl}}
  D_u \partial_{\xi\xi} u_{1,j}  &=& - F^\rms(\ust_j,v_{0,j})v_{0,j} - F^\rmf(\ust_j,v_{0,j})[u_{1,j}, v_{0,j}]\\
  D_v \partial_{\xi\xi} v_{0,j} &=& - c_j \partial_\xi v_{0,j} - G^\rms(\ust_j,v_{0,j})v_{0,j} - G^\rmf(\ust_j,v_{0,j})[u_{1,j},v_{0,j}]\\
  \displaystyle\lim_{\xi\to\pm\infty} v_{0,j}(\xi) &=& 0\\
  \displaystyle\lim_{\xi\to\pm\infty}\partial_\xi u_{1,j}(\xi) &=& \partial_\xi U_0(r_j\pm),
\end{array}
\end{equation}
where $u_{0,j}\equiv\ust_j$ and this typically vanishes; else it must be another common root of $F^\rmf, G^\rmf$ and thus lies in a discrete set.
Note that the case $F^\rmf\neq 0, G^\rmf=0$ either generates uncoupled interface motion or velocities as in the second order case that are determined by a higher order expansion, which we omit. Note that 
$F^\rmf = G^\rmf \equiv 0$, $c_j=0$ in \eqref{e:sest-gen} gives \eqref{e:sndinter}, that is, second order interaction. Recall that for fronts the nonlinearities have roots in $V$  at $V^\front$. 

\medskip
An interface is a solution of \eqref{e:sest-gen} that is homoclinic (pulse) or heteroclinic (front) to the invariant manifolds at $v_0\in\{0,V^\front\}$, which consist of affine $u_1$. Compare \eqref{e:linsemi-small3}. To leading order, matching involves only the nearest left and right neighbors (and $\ust_j$ is typically independent of $j$), so that  $R_j(r_1,\ldots,r_N) = R_j(r_{j-1},r_j,r_{j+1})$. 
%
It is useful to view \eqref{e:sest-gen} as an a priori description of the velocities $c_j = C(P_j^-, P_j^+,\ust_j)$ with $P_j^\pm =  \partial_\xi U_0(r_j\pm)$ as parameters. See \S\ref{s:corenum} for numerical computations of an example case. In terms of $C$ and for $\tau= \eps t$ the leading order equations of motion for first order interaction are
\begin{equation}
\begin{array}{rcl}\label{e:sestr-dyn-gen}
D_u \partial_{xx}U_0 &=& - \H(U_0)\;,\quad x\in \cup_{j=0,{N}}(r_j,r_{j+1})\\
U_0(r_j) &=& \ust_j \;,\quad j=1,\ldots,N\\
\frac{\rmd}{\rmd\tau} r_j &=& C\left( \partial_\xi U_0(r_j-),  \partial_-\xi U_0(r_j+), \ust_j \right), 
\end{array}
\end{equation}
where $\H(U_0) = H(U_0,0;0) - F_*^\rmf(U_0,0)$ for certain $F_*^\rmf$ (see \S\ref{s:standard}), and boundary conditions apply at $r_0, r_{N+1}$. 

\bigskip
A benefit of model independent equations of motion is that for specific equations, where numerical observation suggest semi-strong interaction, the laws of motion can be readily computed and compared with the simulations. Moreover, general properties can be identified a priori. Assuming solvability of the first order equations of motion, $n=1$ and certain properties of $J$ and $C$, we prove in \S\ref{s:fast} that first order pulse motion possesses various Lyapunov functionals: in particular, the largest amplitudes of $U$ and distances between pulses decay. This severely constrains the possible leading order dynamics. For instance, periodic interface motion is not possible. 


\section{Standard forms for semi-strong interaction}\label{s:standard}
In this section we derive \eqref{e:stand} and the mentioned properties from \eqref{e:rds} under Hypothesis~\ref{h:main}. In particular, we derive the leading order small and large scale equations. As a first step, we note that non-constant $U_0$ and constant $V_0$ away from interfaces implies that $G^\rms, G^\rmf$ in \eqref{e:stand}$_2$ must have a factor $V-V^*$ for $V^*=V^\base=0$ in case of pulse and $V^*\in\{0, V^\front\}$ in case of fronts. 

This a priori rules out semi-strong interaction for equations with $E=0$ and $G$ linear in $U$, such as the aforementioned FitzHugh-Nagumo model, where $G(U,V) = \rho (U-\gamma V)$ (see \cite{FHN}), and the Oregonator model, where $U=(w_1,w_2)$ and $G(U,V) = w_1-V$ (see \cite{oreg}). 

\medskip
In the following we consider solutions $U_\eps, V_\eps$ from Hypothesis~\ref{h:main} and omit the subindex $\eps$ for brevity. Hypothesis~\ref{h:main} implies that, away from interfaces, 
%
%
%
%
%
\begin{equation}\label{e:rds2}
\begin{array}{*{1}{rcl}}
D_u \partial_{xx}U + F(U,V;\eps) &=& \calO(\eps)\\
\eps^2 D_v \partial_{xx} V +G(U,V;\eps)V + \eps E(U,V; \eps) &=& \calO(\eps),
\end{array}
\end{equation}
where the right hand side contains the time derivatives. 
As before, we write $U,V$ on the $x$-scale and $u,v$ on the $\xi$-scale. We repeatedly make use of the smoothness assumption in Hypothesis~\ref{h:main} without explicit mentioning.

%
%

\medskip
By Hypothesis~\ref{h:main}, we can write $r_j(t) = r_j(\tau,\eps)$ with $\tau=\eps t$. Setting $c_j = \frac{\rmd}{\rmd\tau} r_j(0,0)$, this means $\frac{\rmd}{\rmd t}|_{t=0} \xi_j = - c_j + \calO(\eps)$. We consider a single interface and drop the subindex $j$ in the following. Substituting the small scale ansatz \eqref{e:rds} and absorbing $\eps E$ into $\calO(\eps)$ gives
\begin{equation}\label{e:smallscale}
\begin{array}{*{1}{rcl}}
D_u \partial_{\xi\xi}u + \eps^2 c \partial_\xi u + \eps^2 F(u,v;\eps) &=& \calO(\eps^3)\\
D_v \partial_{\xi\xi} v + c \partial_\xi v + G(u,v;\eps)v &=& \calO(\eps).
\end{array}
\end{equation}
Pulses are solutions of these equations in the limit $\eps=0$ whose $v$-components are bi-asymptotic to $\{v=0\}$, and fronts heteroclinic connections between $\{v=0\}$ and $\{v=V^\front\}$,

In preparation for a subsequent matching of slopes on large and small scales, as in \eqref {e:fstslopes}, we write \eqref{e:smallscale} as a first order system in the form
\begin{equation}\label{e:smallscale-fst}
\begin{array}{*{1}{rcl}}
\partial_\xi u &=& \eps p\\
D_u \partial_\xi p &=& - \eps c p - \eps F(u,v;\eps) + \calO(\eps^2)\\
\partial_\xi v &=& q\\
D_v q_\xi &=& - c q - G(u,v;\eps)v + \calO(\eps).
\end{array}
\end{equation}
Hypothesis~\ref{h:main} yields an expansion $u= u_0 + \calO(\eps)$, $p= p_0 + \calO(\eps)$, $v= v_0 + \calO(\eps)$, $q= q_0 + \calO(\eps)$, $u=u_0+\calO(\eps)$. Setting $\eps=0$ in \eqref{e:smallscale-fst} gives constant $u_0, p_0$, which means that the amplitudes and slopes, respectively, of $U_0$ to the left and right of an interface are equal to leading order. Since \eqref{e:rds2}$_1$ yields a second order ordinary differential equation for $U_0$, it follows that a \emph{pulse} interface does not affect $U_0$. It can thus be any fixed solution to \eqref{e:rds2}$_1$ at $\eps=0$ that satisfies the boundary conditions. Hence, as in \S\ref{s:bex-pulse}, pulses would move to leading order in a fixed external field given by the large scale solution. However, at a \emph{front} interface the value of $V$ in \eqref{e:rds2} jumps between $V^\base$ and $V^\front$ and it thus affects the large scale solution as in \S\ref{s:bex-front}.

\medskip
The decoupling for pulses disappears for $F$ of the form $F(u,v;\eps) = \eps^{-1}\tF(u,v;\eps)$. It turns out that the exact nature of the singularities of $F$ and $G$ is essential for the type of semi-strong interaction. Recall that here we assume boundedness of $U,V$ and arrive at the conclusion that the right hand side of \eqref{e:rds} must have a singular term. In contrast, in \S\ref{s:bex-pulse}  we scaled $v$ a posteriori to obtain boundedness and thereby introduced a singularity in the originally regular right hand side.

\medskip
For clarity of the exposition, we now focus on \emph{pulses} so that $V_0\equiv 0$. For fronts all requirements at $V=V^\base$ equally apply to $V=V^\front$ and are explicitly noted in \S\ref{s:main}.

\subsection{Second order semi-strong standard form}\label{s:second}

Let us first consider the lowest order singularity 
\[
F(U,V;\eps) = H(U,V;\eps) + \eps^{-1} F_1(U,V),
\]
and regular $G$ so that $G(U,V;\eps) = G^\rms(U,V)$ without loss of generality, by modifying $E$, if required.

In order to solve on the large scale, we need to evaluate $F$ at $\eps=0$. This is possible if $F_1(U,V) =  F^\rms(U,V)V$, because by Hypothesis~\ref{h:main} we have\footnote{For fronts $V-V^\front = \calO(\eps)$ and $F^1$ has a factor $V-V^\front$. We omit similar comments below.} $V=\calO(\eps)$ for $x\neq r_j$. Specifically, $V_1$ enters into the equation for $U_0$ given by
\begin{equation}\label{e:U0}
\partial_t U_0 = H(U_0,V_0;0) + F^\rms(U_0,V_0)V_1,
\end{equation}
so that the assumption  $\partial_t U_0=\partial_t V_0 \equiv 0$ implies $\partial_t V_1\equiv 0$ or $F^\rms(U,0)\equiv 0$. In the non-trivial former case the equation for $V_1$ reads
\begin{equation}\label{e:V1lg2nd}
0 = G^\rms(U_0,V_0)V_1 + E(U_0,V_0;0),
\end{equation}
and solvability guaranteed by Hypothesis~\ref{h:main} implies that $E$ vanishes if $G^\rms$ does. Substitution of the solution into \eqref{e:U0} gives
\begin{equation}\label{e:slowlarge}
D_u \partial_{xx}U_0 + H(U_0,0;0) - F_*^\rms(U_0,V_0) E(U_0,0;0) = 0,
\end{equation}
with suitable $F_*^\rms(U_0,V_0)$. 
In fact, $F_*^\rms$ vanishes in all cases treated in the literature, where $E\equiv 0$ or $F^\rms\equiv 0$, or $F^1$ is quadratic in $V$. Compare \S\ref{s:apex}.

\medskip
The leading order small scale problem \eqref {e:smallscale-fst} now has the form
\begin{equation}\label{e:gen2}
\begin{array}{*{1}{rcl}}
  \partial_{\xi} u_0 &=& 0\\
  D_u \partial_\xi p_0 &=& - F^\rms(u_0,v_0) v_0\\
  \partial_{\xi}v_0 &=& q_0\\
  D_v \partial_\xi q_0 &=& - cq_0 - G^\rms(u_0,v_0)v_0.
\end{array}
\end{equation}

In order to obtain a complete leading order solution, the large and small scale solutions need to match appropriately at the interface for $x=r_j$. For pulses $V_0= V^\base=0$ on both sides of the interface and $u_{0}$ is constant so that continuity of $U_0$ and $V_0$ requires $u_{0}(\xi) \equiv U_0(r_j)$ and $v_{0}(\xi)\to 0$ as $\xi \to\pm\infty$. For the derivatives of $U$ and $u$, matching means to leading order (note $v_0=v_{0,j}, p_0=p_{0,j}$)
\begin{equation}\label{e:largeslope}
 \partial_x U_0(r_j\pm)
= p_0(\pm\infty) =  p_{0}(0) -\int_0^{\pm\infty} D_u^{-1} F^\rms(U_0(r_j),v_{0}(\xi)) v_{0}(\xi) \rmd\xi.
\end{equation}

Pulse interfaces thus correspond to solutions of \eqref{e:gen2} that are homoclinic to $V^\base=0$ in the $(v_0,q_0)$-equations of \eqref{e:gen2}. 
The situation for fronts is analogous. See \S\ref{s:simpsym} for the effect of symmetries in the laws of motion.

\medskip 
Concerning Remark~\ref{r:sing} note that while the large scale problem allows for general $\eps^{-j}$-terms in the $U$-equation if $F$ has a factor $V^j$, the small scale problem structure necessarily changes for $j>1$ and leads to first order interaction as discussed next. 


\subsection{First order semi-strong standard form}\label{s:first}
Let us now consider the case of other singularities in $F, G$. From the above discussion, first order semi-strong interaction arises if the reversible symmetry is broken so that $c=\calO(1)$ is typically required to locate solutions. 

\subsubsection{Small scale problem}\label{s:first-small} The next order for singularities is $F=\calO(\eps^{-2})$ or $G=\calO(\eps^{-1})$. We start out by considering $F$ and suppose there is smooth $\tilde F(U,V;\eps)$ such that
\begin{equation}\label{e:Fform1}
F(U,V;\eps)= \eps^{-2}\tilde F(U,V;\eps).
\end{equation}
Expanding $\tilde F$ in $\eps$ and arguing as in the second order case this means that 
\begin{equation}\label{e:Fform2}
\tilde F(U,V;\eps) = F_2(U,V)V + \eps F^\rms(U,V)V + \eps^2 H(U,V;\eps).
\end{equation}
From \eqref{e:smallscale-fst}$_1$ it follows that $\partial_\xi u_0 =0$ so that $u_0\equiv\ust$ is constant. Now the expansion $u = u_0 + \eps u_1 + \calO(\eps^2)$ in \eqref{e:smallscale-fst} gives 
\begin{equation}\label{e:fast1}
\begin{array}{*{1}{rcl}}
   \partial_\xi u_1 &=& p\\
  D_u \partial_\xi p &=&- F^\rms(u^\rmf + \eps u_1,v)v -
  \eps^{-1} F_2(u^\rmf + \eps u_1,v)v + \calO(\eps)\\
  \partial_{\xi}v &=& q\\
 D_v \partial_\xi q &=& - c q - G(u^\rmf + \eps u_1,v;\eps)v  + \calO(\eps).
\end{array}
\end{equation}
Since $v_0$ is non-constant in the interface, to compensate the singular coefficient of $F_2$ requires a root at $u_0\equiv\ust$ so that $F_2(u,v) = F^\rmf(u,v)(u-\ust)$. This implies that at \emph{any} interface location the value of $u_0$ must be at a root of $F^\rmf$. Without loss of generality, by shifting $u$, we can assume $\ust=0$ at least for one interface. 

\medskip
Setting $\eps=0$ in the equation for $q$ in \eqref{e:fast1} generates a right hand side that is independent of $u_1$ and $p_0$, so that the symmetry argument from the second order case applies, which means interaction with motion of order $\eps^2$ -- at least for symmetric pulses. Hence, generally we need to allow for a term of order $\eps^{-1}$ in $G$, and, as for $F$, this must have a factor $(u-\ust)$. In particular, $\ust$ must be a simultaneous root of $F^\rmf$ and $G^\rmf$.  Indeed, in the example \eqref{e:linsemi-small2} the unique root of $G$ and $F$ in $u$ lies at zero. Since terms of order $\eps$ can be absorbed into $E$, we obtain
\[
G(u,v;\eps) = G^\rms(u,v) + \eps^{-1} G^\rmf(u,v)(u-\ust).
\]
Substitution into \eqref{e:fast1} gives to leading order
\begin{equation}\label{e:fst}
\begin{array}{*{1}{rcl}}
  \partial_\xi u_1 &=& p_0\\
  D_u \partial_\xi p_0 &=& - F^\rms(\ust,v_0)v_0 - F^\rmf(\ust,v_0)[u_1, v_0]\\
  \partial_{\xi}v_0 &=& q_0\\
  D_v \partial_\xi q_0 &=& - c q_0 - G^\rms(\ust,v_0)v_0 - G^\rmf(\ust,v_0) [u_1,v_0],
\end{array}
\end{equation}
where only the root at $u_1 = \ust=0$ is explicitly noted.

This is the generalisation of \eqref{e:linsemi-small2} and analogously $\{v_0=\partial_\xi v_0 = 0\}$ is an invariant subspace which consists of affine $u_1(\xi) = \const_1\xi + \const_2$ for any $\const_1, \const_2 \in \R^n$. This space is also a center manifold of any equilibrium with $v_0 = 0$, and for $c=0$ its transverse eigenvalues are the square roots of those of the matrices $G^\rms(\ust,0)+G^\rmf(u_1,0)[u_1,\cdot]$.

Matching large and small scale requires $p_{0}(\pm\infty)=\partial_x U_0(r_j\pm)$, which gives the claimed laws of motion \eqref{e:sest-gen}. In contrast to the second order small scale problem \eqref{e:gen2}, here $\partial_\xi p_0$ depends on $u_1$ to leading order. In particular, asymmetry of $\partial_x U_0$ at $r_j$ implies asymmetric boundary conditions $p_{0}(-\infty) \neq -p_{0}(\infty)$, which break the reversible symmetry at $c=0$. Therefore, typically $c\neq 0$ is required to locate a solution that is homoclinic or heteroclinic in the $v_0$-component. This generates first order interaction.

\subsubsection{Large scale problem}\label{s:fst-lg}
Since $U_0\not\equiv 0$ and $V_0\equiv 0$ in Hypothesis~\ref{h:main}, regularity of \eqref{e:rds2}$_1$ with $F(U,V;\eps)$ as in \eqref{e:Fform1} and \eqref{e:Fform2} at $\eps=0$ implies $V_1\equiv 0$ or $F^\rmf$ has a factor $V$. We refer to the former as the linear case and latter as the quadratic case, which occurs in the examples in \S\ref{s:bex-pulse} and \S\ref{s:apex} below, where $V_1=V_1^-\not \equiv 0$.

On the large scale \eqref{e:rds2}$_2$ (multiplied by $\eps$) is of the form
 \begin{equation}\label{e:fst-lg-v}
\eps^2 E(U,V;\eps) +  \eps G^\rms(U,V)V + G^\rmf(U,V)[U,V] = \calO(\eps^2),
\end{equation}
so that analogously $V_1\equiv 0$ or $G^\rmf$ has a factor $V$. 

\paragraph{The linear case}
Recall that the right hand side of \eqref{e:fst-lg-v} is $\eps\partial_t V$, so that $V_0\equiv V_1\equiv0$ means it is $\calO(\eps^3)$. Therefore,  \eqref{e:fst-lg-v} at order $\eps^2$ yields an algebraic solvability condition (solvable by Hypothesis~\ref{h:main}) with $E=0$ whenever $G^\rmf=0$. Substituting the corresponding solution $V_2$ back to the $U$-equation, gives the leading order large scale problem 
\begin{equation}\label{e:fastlarge}
D_u \partial_{xx}U_0 + H(U_0,0;0) - F_*^\rmf(U_0,V_0) = 0,
\end{equation}
with suitable $F_*^\rmf$ analogous to \eqref{e:slowlarge}. The correction term $F_*^\rmf$ only depends on $F^\rmf$, $G^\rmf$ and $E$ at $V=0$, and vanishes if either of these does. 

\paragraph{The quadratic case} Here $\partial_t U_0\equiv 0$ in \eqref{e:rds2}$_1$ implies $\partial_t V_1\equiv 0$ so that, as in the linear case, the right hand side of \eqref{e:fst-lg-v} is $\calO(\eps^3)$. The arising algebraic solvability condition reads
\[
E(U_0,V_0;0) + G^\rms(U_0,V_0)V_1 + G^\rmf(U_0,V_0)[U_0,V_1^2]=0,
\] 
where $G^\rmf$ has been redefined according to the additional factor $V$ required above for $V_1\neq 0$. While $V_1=0$ is one solution for $E\equiv 0$, nontrivial $G^\rms(U_0,0)$ and $G^\rmf(U_0,0)$ generate another solution. (The assumed boundedness in $x$ may after all select $V_1=0$ as in \S\ref{s:apex} below.) Hence, in general, the correction $F^\rmf_*$ in \eqref{e:fastlarge}, also depends on $G^\rms$. Compare \S\ref{s:apex}.

\subsection{Further comments}\label{s:disc}

Since $U_{0,j}$ obeys a second order ODE for all $j$, and $U_{0,j}^\pm(r_j)=u_{0,j}=\ust_j$ lies at roots of $(F^\rmf, G^\rmf)$ it follows that the limiting slopes $\partial_x U^\pm(r_j\pm)$ completely determine the adjacent large scale solutions $U_{0,j}, U_{0,j+1}$. The matching problem is therefore local in the interface sequence to leading order, and the velocities $c_j$ can be viewed as a function of the left and right slopes (and roots $u^\rmf_j$ of $(F^\rmf, G^\rmf)$) alone:
$c_j=C(\partial_x U^-(r_j-), \partial_x U^+(r_j+), \ust_j)$.
The laws of motion for first order interface dynamics in terms of $C$ are given by \eqref{e:sestr-dyn-gen}.
However, even for the simplest cases, nothing is known analytically about $C \neq 0$. Compare \S\ref{s:apex}. Nevertheless, for $n=1$, the local coupling and $u_0\equiv u^\rmf$ have strong consequences for possible interface dynamics. See \S\ref{s:fast}.



As mentioned in \S\ref{s:main}, the `normal form' proposed in \cite{DoeKap} for pulses and the three component model in \cite{HeDoeKap1} with fronts have the second order standard form for semi-strong interaction, and indeed the relative motion has velocities of order $\eps^2$. See \cite{DoeKap, WaSuRu,EhWo, HDKP}. The present abstract derivation explains why for first order interaction a factor $UV$ (and not just $V$) is required in the nonlinear kinetics: there must be a simultaneous root in $U$ of the leading $U$- and $V$-kinetics. Moreover, the special role of the quadratic nonlinearities $UV^2$ in both $F$ and $G$ that arises in the `normal form' of \cite{DoeKap} also occurs on the present abstract level.

In case $E\equiv 0$, the corrections $F_*^\rmf$ and (typically) $F_*^\rms$ vanish and so the large scale vector field is the same for first and second order interaction (as in the examples from \S\ref{s:bex-pulse}). Thus it does not reveal the order of semi-strong interaction.

\medskip 
Concerning Remark~\ref{r:sing}, as in the second order case  the large scale problem allows for general $\eps^{-j}$-terms in the $U$-equation if $F$ has a factor $V^j$ (or if $V$ is order $\eps^j$). However, here the small scale problem is consistent with this, if $F$ also possesses a factor $U^{j-1}$. For the $V$-equation the same applies to $G$ with powers incremented by one.

\section{Application to examples}\label{s:apex}

In this section we apply the standard forms to immediately see how first and second order semi-strong pulse interaction arise in the class of models given by
\begin{equation}\label{e:ex}
\begin{array}{rl}
\partial_t U &= \partial_{xx}U + \alpha - \mu U + \gamma V - \rho UV^2\\
\partial_t V &= \eps^2 \partial_{xx}V  + \beta - V + UV^2,
\end{array}
\end{equation}
where $U, V$ are scalar and $\alpha,\beta,\gamma,\mu,\rho$ are parameters. This is a combination of the model from \S\ref{s:bex-pulse} ($\rho=\mu=\beta=0$), the Schnakenberg model ($\mu=\gamma=0$ \cite{KoWaWe2}), 
the Gray-Scott model ($\alpha=\mu$, $\gamma=\beta=0$ \cite{DoeKap,MurOsi2,Rey}) and the Brusselator model ($\alpha=\mu=0$ \cite{KoErWe}). 

In order to see which type of semi-strong interaction occurs, we substitute the scalings 
\[
U = \eps^r \tU, \quad V = \eps^s \tV
\]
into \eqref{e:ex}. Upon dividing by $\eps^r$ and $\eps^s$, respectively, this gives
\begin{equation}\label{e:ex-scaled}
\begin{array}{rl}
\partial_t \tU &= \partial_{xx}\tU + \eps^{-r}\alpha - \mu \tU + \gamma \eps^{s-r}\tV - \eps^{2s}\rho \tU\tV^2\\
\partial_t \tV &= \eps^2 \partial_{xx}\tV  + \eps^{-s}\beta - \tV + \eps^{s+r}\tU\tV^2.
\end{array}
\end{equation}

\subsection{Second order semi-strong interaction}
Comparing \eqref{e:ex-scaled} with \eqref{e:stand}, we infer $F^\rmf = G^\rmf=0$ requires $s-r \geq -1$ and $2s \geq -1$ with at least one equality, and $r+s=0$. Both cases yield $s=-r=-1/2$, which gives
\begin{equation}\label{e:ex-second}
\begin{array}{rl}
\partial_t \Uw &= \partial_{xx}\Uw + \eps^{-1/2}\alpha - \mu \Uw + \eps^{-1} \gamma\Vw - \eps^{-1}\rho \Uw\Vw^2\\
\partial_t \Vw &= \eps^2 \partial_{xx}\Vw  + \eps^{1/2}\beta - \Vw + \Uw\Vw^2.
\end{array}
\end{equation}
It follows that second order interaction requires $\alpha =\eps^{1/2}\alw$ with bounded $\alw$ and $\beta =\eps^{1/2}\bew$ with bounded $\bew$ for a regular expansion, and with $V_0=0$ in mind. The small scale problem \eqref{e:gen2} applied to this case reads 
\begin{equation}\label{e:smallex}
\begin{array}{rl}
\partial_\xi \uw_0 &= 0\\
\partial_\xi p_0 &= - \gamma\vw_0 + \rho \uw_0\vw_0^2\\
\partial_\xi \vw_0 &= q_0\\
\partial_\xi q_0 &= - cq_0 + \vw_0 - \uw_0\vw_0^2.
\end{array} 
\end{equation}
Note that the $(v,q)$-equations are the same as \eqref{e:linsemi-small1}; in particular only pulses exist.

On the other hand, expanding \eqref{e:ex-second} in $\eps$ gives $\Vw_0=0$ and $\Vw_1=\bew$ so that the large scale problem \eqref{e:slowlarge}, which is obtained by expanding \eqref{e:ex-second}$_1$ reads
\begin{align}\label{e:ex-large}
\partial_{xx}\Uw_0 &= -(\alw + \gamma\bew) + \mu \Uw_0,
\end{align}
generalizing \eqref{e:linsemi-lg}. 


\subsection{First order semi-strong interaction}
The comparison of \eqref{e:ex-scaled} with \eqref{e:stand} shows that first order interaction requires $s-r \geq -2, 2s \geq -2$ with at least one equality and $s+r=-1$. Both cases yield $r=0$, $s=-1$ and thus
\begin{equation}\label{e:fst-apex}
\begin{array}{rl}
\partial_t U &= \partial_{xx}U + \alpha - \mu U + \eps^{-1}\gamma \Vs - \eps^{-2}\rho U\Vs^2\\
\partial_t \Vs &= \eps^2 \partial_{xx}\Vs + \eps\beta -\Vs + \eps^{-1} U\Vs^2.
\end{array}
\end{equation}
Since $F^\rmf$ has a factor $\Vs$ this is the quadratic case from \S\ref{s:fst-lg}. We thus expand \eqref{e:fst-apex}$_2$ in $\eps$ and solve for $\Vs_1$, which suggests the two solutions
$\Vs_1^\pm = \frac 1 {2U_0} \pm \sqrt{\frac 1 {4U_0^2} - \frac{\beta}U_0}$.
However, the fact that $U_0=0$ at pulse locations implies that only $\Vs_1=\Vs_1^-$ is an option.
Substitution into the expansion of \eqref{e:fst-apex} yields, in contrast to \eqref{e:ex-large}, the large scale problems
\begin{align}\label{e:fst-lg}
\partial_{xx}U_0 &= -\alpha + \mu U_0 - \gamma \Vs_1^- - \rho U_0\left(\Vs^-_1\right)^2.
\end{align}

Now we turn to the small scale problem. Application of \eqref{e:fst} to \eqref{e:fst-apex} gives
\begin{equation}\label{e:core}
\begin{array}{rl}
\partial_{\xi\xi} u_1 &= \rho u_1\vs_0^2 - \gamma \vs_0\\
\partial_{\xi\xi} \vs_0 &= - c\partial_{\xi} \vs_0 + \vs_0 - u_1\vs_0^2,
\end{array}
\end{equation} 
In analogy with \eqref{e:linsemi-small3}, for any $c$ the invariant subspace $\{\vs_0=\partial_\xi \vs_0 =0\}$ of \eqref{e:core} consists of affine $u_1(\xi) = \const_1 \xi+ \const_2$, $\const_1,\const_2\in\R$. Compare also  \eqref{e:fst}. For the four dimensional flow of \eqref{e:core}, this space is the two-dimensional center manifold of any equilibrium with $\vs_0=0$ and has transverse eigenvalues $\left(-c\pm\sqrt{4+c^2}\right)/2$. Hence, it is normally hyperbolic with one-dimensional stable and one-dimensional unstable manifolds and the center manifold has the same properties as in \eqref{e:linsemi-small3}.

\medskip
For $\gamma=0, \rho=1$ \eqref{e:core} is the well-known `core problem,' first derived and numerically analysed in \cite{Rey}, see also \cite{MurOsi2}. In this case also $u_1=0$ is invariant and $u_1$ convex whenever $u_1>0$. For $c=0$, the existence of various even pulse solutions which are homoclinic in $\vs_0$ and asymptotically affine in $u_1$ was proven in \cite{DoeKapPel}. In \S\ref{s:corenum} we provide a more detailed numerical study, which also shows that asym\-metric pulse-type solutions require $c\neq 0$. Indeed, in contrast to \eqref{e:smallex} here matching requires $\lim_{\xi\to\pm\infty} u_1(\xi) = \partial_x U_0(r_j\pm)$, which breaks the reversible symmetry if $\partial_x U_0(r_j-) +  \partial_x U_0(r_j+) \neq 0$. Compare \S\ref{s:first-small}. Therefore, motion is expected to be of order $\eps$, though nothing is known rigorously.

\medskip
Note that here second order interaction is an asymptotic regime within first order interaction. Just as for the example in \S\ref{s:bex-pulse}, substituting the scaling for second order interaction into \eqref{e:fst-lg} and \eqref{e:core} yields \eqref{e:ex-large} and \eqref{e:smallex}, respectively.

\section{Laws of motion for second order semi-strong interaction}\label{s:gen-se-we}

In this section we prove Proposition~\ref{p:sndlaw}, which gives the implicit form of second order interaction laws and simplify these further in case of symmetries.

\subsection{Proof of Proposition~\ref{p:sndlaw}} 
On the small scale, \eqref{e:stand} for $G^\rmf=F^\rmf=0$ reads
\begin{align}
 \eps^{2}\partial_t u &=D_u \partial_{\xi\xi}u +  \eps^{2}H(u,v;\eps) + \eps F^\rms(u,v)v \label{e:uslowsmall}\\
\partial_t v &= D_v \partial_{\xi\xi}v + G^\rms(u,v)v + \eps E(u,v;\eps). \label{e:vslowsmall}
\end{align} 
In the following, solutions from Hypothesis~\ref{h:main} are considered, and for readability we suppress dependencies of $u, v$ on $j$ and $\eps$. Hypothesis~\ref{h:main} allows to expand \eqref{e:uslowsmall}, \eqref{e:vslowsmall} in $\eps$ at $\eps=0$ so that that terms of equal order in this expansion must coincide.

\medskip 
Order $\eps^{0}$ in \eqref{e:uslowsmall} gives the condition $\partial_{\xi\xi}u_0 = 0$ as in \eqref{e:gen2}, which implies  $u_0 \equiv a_j\in\R$, and in \eqref{e:vslowsmall} we recover the $(v_0,q_0)$-equations from \eqref{e:gen2} for $c=0$ as in \eqref{e:sndinter}.  

At order $\eps$ equation \eqref{e:uslowsmall} yields 
\begin{equation}\label{e:ueps}
D_u \partial_{\xi\xi} u_1 + F^\rms(u_0(r_j),v_0)v_0 = 0,
\end{equation}
which gives \eqref{e:sewe-sol-u} and determines $u_1$ by $u_0,v_0$ up to affine terms. These are fixed by coupling to neighboring interfaces or the boundary conditions. 

\medskip
Towards order $\eps$ in \eqref{e:vslowsmall}, first linearize the right hand side of \eqref{e:vslowsmall} in an interface pattern from Hypothesis~\ref{h:main}. At order $\eps$, this yields the linear operator
\[
\calL :=  D_v \partial_{\xi\xi} + G^\rms(u_0,v_0)+ \partial_v G^\rms(u_0,v_0)v_0,
\]
and $\partial_\xi v_0 \in \ker(\calL)$ due to translation symmetry in $\xi\in\R$. The right hand side of \eqref{e:uslowsmall} at order $\eps$ is the left hand side of \eqref{e:ueps}, which is independent of $v_1$. Hence, the linearization of \eqref{e:uslowsmall}, \eqref{e:vslowsmall} is block-triangular at order $\eps$ and the $N$ $\calO(\eps)$  eigenvalues from Hypothesis~\ref{h:main} stems precisely from the  translations of interfaces. Substituting $\xi = (x-r_j)/\eps$, 
the left hand side of \eqref{e:vslowsmall} therefore gives
\begin{equation}\label{e:LHS-slow}
\frac{\rmd}{\rmd t}v(t,\xi) = - \eps \partial_\xi v_0 \, \frac{\rmd}{\rmd \tau}  r_{0,j} + \calO(\eps^2). 
\end{equation}

At an interface $(u_0,v_0)$ the right hand side of \eqref{e:vslowsmall} expands to 
\[
\eps\left(
\calL v_1 + \partial_u G^\rms(u_0,v_0) [u_1, v_0] + E(u_0,v_0;0)
\right) + \calO(\eps^2).
\]
Comparison with terms of order $\eps$ from \eqref{e:LHS-slow} implies
\[
-\calL v_1 = \partial_\xi v_0 \, \frac{\rmd}{\rmd\tau} r_{0,j}  + \partial_u G^\rms(u_0,v_0) [u_1, v_0] + E(u_0,v_0;0).
\]
Let $w\in \ker\calL^*$, with the $L^2$-adjoint $\calL^*$. Since $u_1$ is determined by \eqref{e:ueps}, we obtain the solvability condition
\[
\langle\partial_\xi v_0 \, \frac{\rmd}{\rmd\tau} r_{0,j}  + \partial_u G^\rms(u_0,v_0) [u_1, v_0] + 
E(u_0,v_0;0), w\rangle = 0,
\]
that determines the velocity by \eqref{e:slowmot1a} and completes the proof. 

(Note that also the conditions encountered in \S\ref{s:standard} on the large scale are recovered by comparing terms of equal order.)


\subsection{Simplifications by symmetries}\label{s:simpsym}

\medskip
All analyses in the literature for second order semi-strong interaction concern models with $n=1$ so that $\calL$ is self-adjoint. Hence, a natural consideration for the present general case is to assume self-adjoint $\calL$, which means $w=\partial_\xi v_0$ can be chosen. In this case we can proceed as follows for symmetric interfaces.

\subsubsection{Pulse} If the interface is a symmetric pulse then $v_0$ is even (for appropriate $v_0(0)$) and $\partial_\xi v_0$ odd. In addition, $E(a_j,v_0;\eps)$ is even as a funtion of $\xi$, and using \eqref{e:sewe-sol-u} with $b_j := p_0(0)$, the function $p_0 - b_j$ is odd. Further, $u_1 = u_1(0) + \int_0^\xi p_0(\zeta) - b_j \rmd \zeta$ is even, while $\int_0^\xi b_j \rmd \zeta = b_j \xi$ is odd. Therefore, when including $j$-dependence for emphasis, \eqref{e:slowmot1a} simplifies to
\begin{equation}\label{e:slowmot3}
 \frac{\rmd}{\rmd \tau} r_j = -\langle \partial_u G^\rms(a_j,v_{0,j})[ 
  b_j\cdot, v_{0,j}], \partial_\xi v_{0,j} \rangle_2/\|\partial_\xi v_{0,j}\|_2^2.
\end{equation}
This is indeed a generalisation of the equations of motion for pulses in the Gray-Scott model \cite{WaSuRu} and the Gierer-Meinhardt model \cite{DoeKap, EhWo}.

The equilibrium $b_1=\ldots=b_N=0$ means that the pulse pattern is stationary to leading order if the average of the slopes $b_j = p_0(0) = (p_0(\infty)+p_0(-\infty))/2$ vanishes at all pulses, for instance at a symmetric configuration where all distances are equal.

\subsubsection{Front} If the interface is a front, then the equations of motion may simplify as follows if there is $v_*$ such that $v_{0}-v_*$ is odd. First, we shift $v$ so that $v_0$ is odd and $V^\base=-V^\front=-v_*$, and $\partial_\xi v_0$ is even. Then $F^\rms$ and $G^\rms$ have the factor $(v-v_*)(v+v^*)$, which is even in $\xi$ for $v=v_0$. If now $F^\rms$ is odd in $v$ then $F^\rms(a_j,v)$ is odd in $\xi$ and, by \eqref{e:gen2}$_2$, $p_0$ is even. With $\beta_j:= u_1(0)$ we then have that $u_1-\beta_j = \int_0^\xi p_0(\zeta)\rmd \zeta$ is odd. Including $j$-dependencies, \eqref{e:slowmot1a} simplifies to
\begin{equation}\label{e:slowmot4}
 \frac{\rmd}{\rmd \tau} r_j = - \langle 
 \partial_u G^\rms(a_j,v_{0,j})[ \beta_j, v_{0,j}] + 
 E(a_j,v_{0,j};0), \partial_\xi v_{0,j} \rangle_2/\|\partial_\xi v_{0,j}\|_2^2.
\end{equation}
As mentioned in \S\ref{s:main}, in \cite{HeDoeKap1} $F^\rms\equiv\partial_u G^\rms\equiv 0$ and the fronts are odd, and indeed the equations of motion correspond to \eqref{e:slowmot4}. 

\subsubsection{Further comments}
For $\partial_uG^\rms\equiv 0$ and $E(U,V;\eps) = e(U;\eps)V$ with scalar $e$, we have $\frac{\rmd}{\rmd \tau} r_j =0$ in \eqref{e:slowmot4} since $\langle v_{0,j},\partial_\xi v_{0,j}\rangle =0$. Hence, such $E$ do not drive second order semi-strong interaction of odd fronts. Clearly, $\partial_u G^\rms\not\equiv 0$ is required for non-trivial \eqref{e:slowmot3}.

\medskip
Through the matching conditions, the parameters $a_j, b_j, \beta_j$ for each $j$ generally depend on  all other $j$. For the three-component model§ from \cite{HeDoeKap1},  the arising globally coupled system of algebraic equations have been derived in \cite{HDKP}. For 
Gierer-Meinhardt models these are contained in 
\cite{EhWo} 
and can have singularities at which the manifold of pulse patterns undergoes a bifurcation.
The vector field for interface motion in general intricately depends on the details of the model and boundary conditions. Numerical observations suggest gradient-like dynamics, which appears to hold true\footnote{P. van Heijster. Personal communication.} for the three-component model of \cite{HeDoeKap1}. However, it seems difficult to prove this in broader generality -- for first order interaction we prove results in this direction below in \S\ref{s:fast}.

\medskip
The degenerate case $F^\rms\equiv 0$ for pulses, which was discussed after \eqref{e:smallscale-fst}, implies that $U_0$ is constant in time and $a_j = U_0(r_j)$, $b_j = \partial_x U_0(r_j)$. Hence, the equations \eqref{e:slowmot3} for each $j$ decouple, and each pulse moves according to the same scalar ODE to leading order. In particular, the reduced pulse motion is monotone, and, if global, each pulse converges to either an equilibrium or infinity.

\section{First order semi-strong pulse motion}\label{s:fast}

In this section we study first order semi-strong motion of \emph{pulses} for \emph{scalar} large scale problems ($n=1$). We assume existence of smooth solutions to the reduced dynamics of \eqref{e:sestr-dyn-gen}. For pulse patterns we make the natural assumption that the states $U(r_j)=\ust_j$ are all equal and thus may be moved to zero, as in \S\ref{s:bex-pulse}; we abbreviate $C(P^-,P^+)=C(P^-,P^+,0)$. 
Thus, pulse positions are Dirichlet boundary conditions for the (leading order) second order large scale problem \eqref{e:fastlarge}, whose solution inbetween pulses (if it exists) is therefore generically determined by the pulse positions alone. 

According to \eqref{e:sestr-dyn-gen} the reduced first order semi-strong dynamics on the time scale $\tau= \eps t$ with interfaces at $r_1<\ldots<r_N$, boundary conditions at $r_0, r_{N+1}\in \R\cup\{\pm\infty\}$, and $\H(U) := H(U,0;0) - F_*^\rmf(U,0) E(U,0;0)$ from \eqref{e:fastlarge} then reads
\begin{equation}
\begin{array}{rcl}\label{e:sestr-dyn}
D_u \partial_{xx}U &=& - \H(U)\;,\quad x\in \cup_{j=0,{N}}(r_j,r_{j+1})\\
U(r_j) &=& 0\\
\frac{\rmd}{\rmd\tau} r_j &=& C\left(\partial_x U(r_j-),\partial_x U(r_j+) \right),
%
\end{array}
\end{equation}
where $j=1,\ldots,N$. For bounded $D$ we assume separated linear or periodic boundary conditions, and in unbounded directions convergence to constant states. Note that $U$ here is $U_0$ in the notation of the previous sections. In this section subindices of $U$, $P$ are not related to the expansion in $\eps$.

The function $C(P^-,P^+)$ is an essential part of the dynamics and effectively parameterises the manifold of pulse patterns. As mentioned, even for basic examples nothing is known analytically about interfaces for $C\neq 0$. For $P^-+P^+=0$ and $C=0$, existence results are given in \cite{DoeKapPel}, but these do not cover the numerically observed folding of $C$ as $P^+$ increases. See \cite{Rey,MurOsi} and Figure~\ref {f:sestr-bif}. 

\medskip
We will show that under suitable assumptions on $C$ and $J$ the first order $N$-pulse motion is gradient-like with respect to various geometrically meaningful Lyapunov-functionals: the largest  interpulse amplitude and distance decay in time, while the smallest of these increase. This severely constrains the leading order pulse dynamics. Before a more abstract analysis, we present some numerical computations of certain $C$.

\subsection{Numerical computations of a first order pulse problem}\label{s:corenum}

Let us reconsider the small scale problem \eqref{e:core} for the two component models \eqref{e:ex}. It determines the pulse velocity and shape for first order interaction when supplied with boundary conditions as in the abstract version \eqref{e:sest-gen}.  In this section we determine 
 $c=C(P^-,P^+)$ by a numerical approach based on continuation in the parameters
 $P^\pm$. For definiteness we focus on the case $\rho=1, \gamma=0$. For $c=0$ and $P^-=-P^+$ this case has been considered by a more ad-hoc numerical approach in \cite{Rey,MurOsi,KoWaWe}. 

The reflection symmetry of  \eqref{e:sest-gen}
\begin{equation}\label{e:sestr-c-sym}
(P^-,P^+,c,\xi) \to (-P^+,-P^-,-c, -\xi),
\end{equation}
implies $C(-P,P) = 0$ for even solutions, and so we adapt parameters to
\[
P^\rma = -(P^+ + P^-), \quad P^\rms = P^+-P^-.
\]
This transforms $C(P^-,P^+)$ to
\[
C_\rms(P^\rms,P^\rma) = C\left(-(P^\rma + P^\rms)/2, (P^\rms - P^\rma)/2\right).
\]
Due to the reflection symmetry, $\xi\to -\xi$ implies $C_\rms(P^\rms,-P^\rma) \to
-C_\rms(P^\rms,P^\rma)$ so that $C_\rms(P^\rms,0)=0$ for even
solutions.

Let us denote the boundary value problem \eqref{e:sest-gen} 
(with nonlinearities as in \eqref{e:core}) compactly by $\calF(c,P^\rms,P^\rma)=0$.  
For the numerical computations we replace the infinite boundary
location of \eqref{e:sest-gen} by $L=100$; changes in $
L$ did not change the results to a noticable degree. We implemented this in the boundary value problem solver and continuation software \textsc{Auto} \cite{auto}, and solved $\calF(c,P^\rms,P^\rma)=0$ along
grid lines in the $(P^\rms,P^\rma)$-plane. In order to fix the location of the pulse we include the interior
condition $\partial_\xi u_1(0) =0$. In the implementation we therefore split the problem into one for
negative $\xi$ and one for positive $\xi$ and couple the resulting 8 equations via continuity
boundary conditions at $\xi=0$ plus $\partial_\xi u_1(0) =0$.

\begin{figure}
\centering
\begin{tabular}{cc}
\scalebox{0.3}{\input{sestr-sym-bif.pstex_t}} 
& \scalebox{0.3}{\input{sestr-cusp.pstex_t}} \\
(a) & (b)
\end{tabular}
\caption{(a) Bifurcation diagram for $P^\rma=0$. The $\vs_0$-component changes shape from unimodal to bimodal at the bullet. The branches with $c\neq 0$ bifurcate from the bimodal branch at $P^\rms_\pf$. (b) Fold curves of $\calF(c,P^\rms,P^\rma)=0$ and corresponding amounts of nearby solutions. The solid curve includes $P_{\rm sn}^\rms$ in (a). The dashed curve is in the region with bimodal pulses and includes the leftmost fold of the branch with $C_\rms<0$ in (a). The corner at $(P^\rms_\#, P^\rma_\#)$ is a cusp singularity.}
\label{f:sestr-bif}
\end{figure}
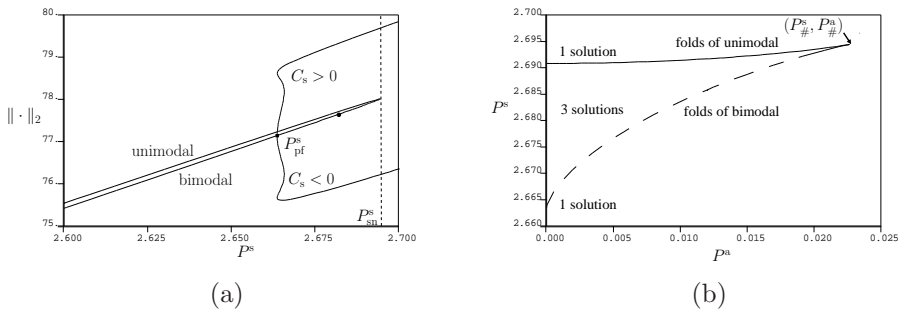

The numerical computations confirm the observation from \cite{Rey,MurOsi} that the relevant branch
of stationary, even, and unimodal solutions 
folds at $P^\rms = P^\rms_\sn\approx 2.69$ and its continuation leads to bimodal (`dimpled') pulses with a local minimum at $\xi=0$. See Figure~\ref{f:sestr-bif}(a). By unimodal we mean that $\vs_0(\xi)$ has a unique critical point at $\xi=0$ (which then is the maximum). For $P^\rms<P^\rms_\sn$ these solutions correspond to stationary pulses solutions of the original problem \eqref{e:ex} and the fold corresponds to a saddle-node bifurcation of these. 

Beyond the results in the literature, we used the symmetric pulses as starting points for a continuation to asymmetric ones, where $P^\rma\neq 0$ thus generating a two dimensional manifold in $(c,P^\rms,P^\rma)$-space. In particular, the fold at $P^\rms = P^\rms_\sn,\; P^\rma=0$ lies on a curve of folds in the $(c,P^\rms,P^\rma)$-space, which is contained in the region $\{P^\rms>P^\rms_\sn\}$. See the solid curve in Figure~\ref{f:sestr-bif}(b).
In addition, these computations for $c\neq 0$ show that the relevant solution set of $\calF(c, P^\rms, P^\rma)=0$ indeed generates a function $C_\rms(P^\rms,P^\rma)$ in the region $P^\rms < P^\rms_\sn$. Specifically, we carefully checked that
\begin{equation}\label{e:sestr-reg}
C_\rms < 0, \quad \partial_{P^\rma} C_\rms <0, \quad \mbox{in\;} \{(P^\rms,P^\rma)| 0\leq P^\rms \leq 2.65, 0 \leq P^\rma \leq 4\},
\end{equation} 
and have not found violations of this in a larger region. Applying the symmetry $C_\rms(P^\rms,-P^\rma) = -C_\rms(P^\rms,P^\rma)$, $\xi\to -\xi$, trivially extends this into the region $P^\rma<0$. 

The extension of this manifold to the region $P^\rms>P^\rms_\sn$ can be understood via the bimodal branch at $P^\rma=0$. See the lower branch in Figure~\ref{f:sestr-bif}(a). At a certain $P_\pf^\rms<P^\rms_\sn$ a pitchfork bifurcation gives rise to a symmetric pair of solution branches with $C_s(P^\rms,0)\neq 0$ of either sign. For $P^\rma\neq 0$ a fold curve emanates from this pitchfork. See the dashed curve in Figure~\ref{f:sestr-bif}(b). This fold curve annihilates with the aforementioned fold curve in a cusp bifurcation at $(P^\rms,P^\rma)=(P^\rms_\#,P^\rma_\#)$. Away from these fold curves, the relevant part of the manifold is a function of $(P^\rms,P^\rma)$.

\medskip
For moderate $\gamma\neq 0$ in \eqref{e:ex}, the bifurcation diagram in Figure~\ref{f:sestr-bif} remains qualitatively the same; the location of the fold is monotone decreasing in $\gamma$. In the case $\gamma=-1, \rho=0$ of the basic example from \S\ref{s:bex-pulse} the diagram has changed more and is plotted in Figure~\ref{f:linear-bif}. Here the pitchfork $P^\rms_\pf$ and fold $P^\rms_\sn$ appear to have merged and no longer generate curves of folds for $P^\rma\neq 0$. 

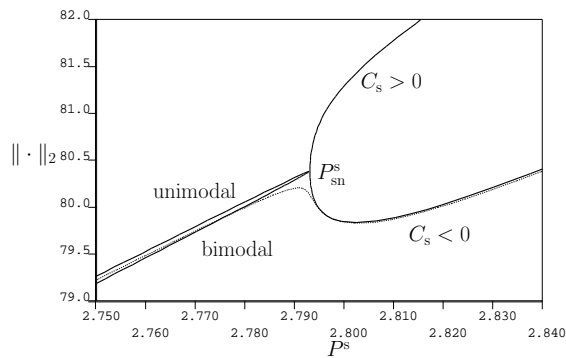
\begin{figure}
\centering
\scalebox{0.4}{\input{linear.pstex_t}} \
\caption{Bifurcation diagram for $\gamma=-1, \rho=0$. Bold lines: $P^\rma=0$, thin dashed line: $P^\rma=0.001$. The fold $P^\rms_\sn$ and pitchfork $P^\rms_\pf$ appear to have merged at a degenerate bifurcation and no longer generate curves of folds.}
\label{f:linear-bif}
\end{figure}

\subsection{Gradient-like nature of first order pulse interaction}

\medskip
We now return to the abstract analysis and Lyapunov-functionals of first order pulse interaction. In preparation, define the following key quantities. 

\begin{definition}\label{d:time}
Let $U(x)$ solve $D_u U_{xx} + J(U)=0$ with initial conditions $U(0)=0$, and $\partial_x U(0)=\tilde P$. Let $\Delta_\pm(\tilde P)\neq 0$ be the smallest positive, respectively largest negative value of $x$ such that $U(x)=0$, and set $\Delta_\pm(\tilde P):=\pm\infty$ correspondingly if there is no such point. 
\end{definition}

Note that due to \eqref{e:sestr-dyn}$_2$, if $U\neq 0$ in $(r_j,r_{j+1})$ then 
\[
r_{j+1}-r_j = \Delta_+(\partial_x U(r_j+))= \Delta_-(\partial_x U(r_{j+1}-)).
\]

\begin{definition}
Let $U(x)$ solve $D_u U_{xx} + J(U)=0$ with initial conditions $U(0)=\tilde U$, and $\partial_x U(0)=0$. Let $\Delta^*(\tilde U)\neq 0$ be the smallest positive value of $x$ such that $\partial_x U(x)=0$, and set $\Delta^*(\tilde U):=\infty$ if there is no such point. 
\end{definition}

We shall show that the aforementioned decay and growth of the largest and smallest amplitudes between pulses occurs, as long as the following hypothesis holds. For the distances a stronger assumption is required; see Corollary~\ref{c:sestr}. Figures~\ref{f:Lyap-per},~\ref{f:Lyap-sep} illustrate the results.

\begin{Hypothesis}\label{h:DN}
For $n=1$ assume that there is $T>0$ such that $\calD_N$ is a non-empty set of pulse positions $r(\tau)$, $r_1(\tau) < \ldots < r_N(\tau)$, which solve \eqref{e:sestr-dyn} for $\tau\in (0,T]$ with bounded $U(\tau,x)$, $x\in D$, that is continuously differentiable in $\tau$ and such that the following hold.
\begin{enumerate}
\item $U(\tau,x)> 0$ for $x\neq r_j(\tau)$,
\item $\Delta^*(\tU)$ is bounded and grows strictly in $\tU$ in a neighborhood of $\max\{U(\tau,x): r_j(\tau)< x < r_{j+1}(\tau)\}$ for all $j=1,\ldots,N-1$, as well as for $j=0$ and $j=N$ if $r_0$ and $r_{N+1}$ are bounded, respectively,
\item $\sgn (C(P^-_j,P^+_j)) = \sgn(P^-_j+P^+_j)$.
%
\end{enumerate}
\end{Hypothesis}

\medskip
A priori the set $\calD_N$ could be empty, but Remark~\ref{r:mono} below and numerical evidence as in Figure~\ref{f:sestr-bif} corroborate that this is not the case for \eqref{e:core} in a broad range of $\rho, \gamma$. However, even then $C(P^-,P^+)$ need not be well-defined globally in time as the dynamics might drive solutions over a bifurcation point of $C$, for instance a fold. 
Another possibility for the PDE dynamics to leave the slow manifold of pulse patterns is a transverse instability of the manifold, beyond which \eqref{e:sestr-dyn} is not meaningful. 

\begin{Remark}\label{r:mono}
The monotonicity in Hypothesis~\ref{h:DN}(2) holds for solutions $U$ of \eqref{e:sestr-dyn} if $\H'(U)\leq0$ and $\H(U)U>0$ for $x\in(r_j,r_{j+1})$. (The proof is given at the end of this section.) This is always the case near saddle points (which means relatively large pulse distance). The known concrete models from \S\ref{s:apex} for first order semi-strong interaction have $\H$ of the form $\H(U) = \const_1 -  \const_2 U$ so that the monotonicity holds for all relevant solutions if $\const_1, \const_2 >0$, and $U < \const_1/\const_2$. 
\end{Remark}

The general monotonicity problem of $\Delta^*$ is closely related to the monotonicity of the period function which has been extensively studied in the literature. See, e.g., \cite{FreireEtAl} and the references therein for a recent account. A fairly practical criterion for monotonicity given in \cite{Chicone} is that $V/\H^2$ be convex.

\begin{Remark}\label{r:sign}
The sign condition $U> 0$ except at pulse positions in $\calD_N$ is not necessary, but holds for the model class in \S\ref{s:apex}. The case $U<0$ can be treated analogously. In fact, if $U$ lies on a periodic orbit, then an even number of sign changes of $U$ between pulses can be removed by shortening $D$ without changing the local dynamics.
\end{Remark}

\begin{figure}
\centering
\scalebox{1}{\input{sestr-arcs-per.pstex_t}} 
\caption{Illustration of the Lyapunov functionals and interface dynamics for periodic $D$ and $N=4$. In this case there are no further local extrema of $\calN(U)$, $\calM(r)$.}\label{f:Lyap-per}
\end{figure}
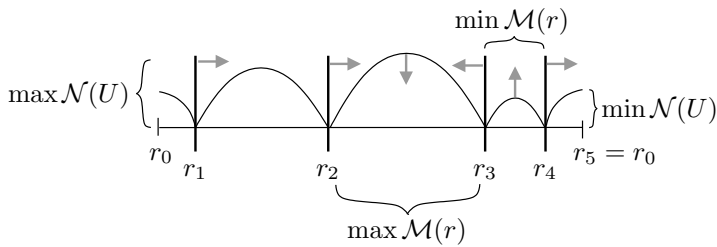

\medskip
We define the set of local maxima of a large scale solution by
\[
  \calN(U) := \left\{ U(x) \,:\, U(x)\mbox{ is a local maximum for }\, x\in D,\, x\neq r_j, \,j=1,\ldots, N  \right\},
\]
and with $r=(r_1,\dots, r_N)$ the distances between pulses by
\[
\calM(r) := \{d_j=r_{j+1} - r_j : j=0,\ldots,N\}.
\] 

For $x\in(r_0,r_1)\cup(r_N,r_{N+1})$, a technical issue is to identify an adjusted maximum and distance that are suitably comparable to those inbetween pulses. Therefore, we first consider periodic $D$ so that $j \in\mathbb{Z}$ mod $N+1$.

\begin{definition}\label{d:mima}
Let $j=0,\ldots, N$.
\begin{enumerate}
\item Let $U_j^*:=\sup\{U(x): x\in (r_j,r_{j+1})\}$. We call $U^*$ a local maximum of $\calN(U)$ if $U^*=U^*_j$ for some $j$ and $U^*> U_{j-1}^*$, $U^*\geq U_{j+1}^*$ or $U^*\geq U_{j-1}^*$, $U^*> U_{j+1}^*$. 
\item We call $d^*$ a local maximum of $\calM(r)$ if $d^*= d_j<\infty$ for some $j$ and if $d^* > d_{j-1}$, $d^* \geq d_{j+1}$ or $d^* \geq d_{j-1}$, $d^* > d_{j+1}$. 
\item Local minima are defined analogously in each case.
\end{enumerate}
\end{definition}

\medskip
The main results are the following. 

\begin{theorem}\label{t:sestr}
Assume Hypothesis~\ref{h:DN}, $n=1$, periodic $D$ and $r=r(\tau)\in\calD_N$ with associated $U=U(\tau,x)$. For each $j=0,\ldots,N$ the following holds.
\begin{enumerate}
\item If $U^*_j$ is a local maximum of $\calN(U)$ then $\frac{\rmd}{\rmd \tau} U^*_j<0$ and $\frac{\rmd}{\rmd \tau} d_j<0$.
\item If $U^*_j$ is a local minimum of $\calN(U)$ then $\frac{\rmd}{\rmd \tau} U^*_j>0$ and $\frac{\rmd}{\rmd \tau} d_j>0$.
\item (i) $\frac{\rmd}{\rmd \tau} U_j^*=0$ $\Leftrightarrow$ (ii) $\frac{\rmd}{\rmd \tau} d_j=0$ $\Leftrightarrow$ (iii)  $U_j^*=U_{j-1}^*=U_{j+1}^*$.
\end{enumerate}
\end{theorem}

\begin{corollary}\label{c:sestr}
Under the assumptions of Theorem~\ref{t:sestr} it holds that $\max \calN(U)$ and $-\min \calN(U)$ are strict Lyapunov functionals in the sense that these are either constant with constant $r$, or the functionals decay strictly in $\tau$ until $r$ lies in the boundary of $\calD_N$ (which may never happen). 
If in addition $\Delta^*$ is strictly monotone increasing on the interval $(\min\calN(U),\max\calN(U))$, then also $\max \calM(r)$ and $-\min \calM(r)$ are Lyapunov functionals in this sense.
\end{corollary}



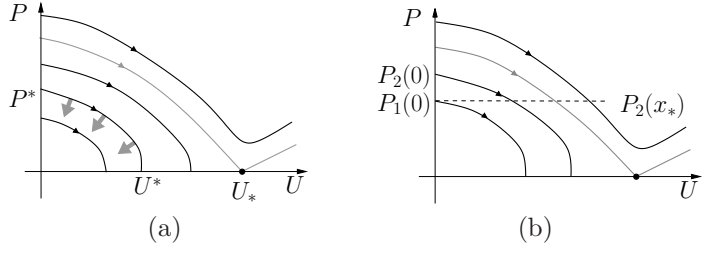
\begin{figure}
\centering
\begin{tabular}{cc}
\scalebox{0.45}{\input{phase1.pstex_t}} 
& \scalebox{0.45}{\input{phase2.pstex_t}}\\
(a) & (b)
\end{tabular}
\caption{Sketch of the phase space geometry. Thin lines are trajectories (separatrices in grey), bullets equilibria. The pictures extend by reflection about the $U$-axis and reversing arrows on trajectories.  (a) For the proof of Theorem~\ref{t:sestr}. The thick gray arrows indicate the dynamics with respect to $\tau$ for a local maximum $U_*\in\calN(U)$. (b) For the proof of Remark~\ref{r:mono}. }\label{f:large}
\end{figure}

\begin{proof}(Theorem~\ref{t:sestr}).
For $n=1$ the scalar large scale problem is reversibly symmetric (from reflection symmetry in $x$) with phase space $(U,P)$, $P= \partial_x U$. For the given solution we denote $(U_j,P_j):=(U,\partial_x U)$ on $[r_j,r_{j+1}]$. 

Since $U_j> 0$ in $(r_j,r_{j+1})$ and $U_j(r_j)=U_j(r_{j+1})=0$, the reversible symmetry about the $U$-axis implies that a unique intersection of $(U_j,P_j)$ with the $U$-axis occurs at $U_j^*\in\calN(U)$. In particular, $d_j=\Delta_+(P_j(r_j+))$ and $P_j(r_j+d_j/2)=0$, which means $U_j^*=U_j(r_j+d_j/2)$.

Such solutions are ordered in the following sense. Consider solutions $\tilde U\geq 0$ to \eqref{e:sestr-dyn}$_1$ with $\tilde U(0)=0$, $\partial_x\tilde U(0)= P^*$ and $\Delta_+(P^*)<\infty$. Since the phase space is two-dimensional, the maxima $U^* = \tilde U(\Delta_+(P^*)/2)$ are \emph{strictly increasing} in $P^*$ as long as $\Delta_+(P^*)<\infty$. See Figure~\ref{f:large}.

\medskip
For $U_j^*$ a local maximum of $\calN(U)$, set $P_j^\pm:= P(r_j\pm)$. The above ordering for interpulse profiles implies that $-P_j^- < P_j^+$ and $-P_{j+1}^-\geq P_{j+1}^+$, or $-P_j^- \leq P_j^+$ and $-P_{j+1}^- > P_{j+1}^+$. By Hypothesis~\ref{h:DN}(3) we thus have $\dot r_j >0$ and $\dot r_{j+1}\leq 0$ or $\dot r_j \geq 0$ and $\dot r_{j+1} < 0$. 
Therefore,
\begin{equation}\label{e:gen-dists}
\frac{\rmd}{\rmd \tau}(r_j - r_{j+1}) >0,
\end{equation}
and monotonicity of $\Delta_+$ in Hypothesis~\ref{h:DN}(2) implies that $\frac{\rmd}{\rmd \tau} U_j^*<0$. 
Analogously we find that $\frac{\rmd}{\rmd \tau} U_j^*>0$ in case of a local minimum, which proves the first two parts of the theorem. 

For the third part, it follows from the above that $\frac{\rmd}{\rmd \tau} U_j^*\neq0$ implies $U_j^*\neq U_k^*$ for $k=j-1$ or $k=j+1$. Since $U_j^*=U_k^*$ implies that $U_j$ and $U_k$ lie on the same trajectory, equivalence of (i) and (iii) follows from Hypothesis~\ref{h:DN}(3). Equivalence of (i) and (ii) is a direct consequence of Hypothesis~\ref{h:DN}(2). \end{proof}

\begin{Remark}
Due to the global ordering of trajectories noted in the proof, there are many more such Lyapunov-functionals: in $(U,P)$-space, the intersection points with any fixed line through the origin in the positive quadrant move towards the origin. See Figure~\ref{f:large}.
\end{Remark}

\begin{proof}(Corollary~\ref{c:sestr}).
The first part immediately follows from Theorem~\ref{t:sestr}. 

Concerning the second part, it follows from the global monotonicity assumption that if $d_j$ is a local extremum of $\calM(r)$, then $U_j^*$ is a local extremum of $\calN(U)$. Now application of the proof of Theorem~\ref{t:sestr} implies the claim.
\end{proof}

\begin{proof}(Remark~\ref{r:mono}).
Generally, $U(r_j)=U(r_{j+1})=0$ and $\sgn(U(x))=\pm 1$ in $(r_j,r_{j+1})$ requires $\sgn(J(U))= \pm1$ at the critical point $U_j^*$. Hence, $J(U)\neq 0$ requires $\sgn(J(U))= \sgn(U)$ for solvability of \eqref{e:sestr-dyn}.
We consider $U, \H(U)> 0$; the negative case follows by a symmetric argument. Hence, we have $\partial_x P = -J(U)<0$ so that $P$ decays strictly in $x$.
 
Let $(U_1, P_1)$, $(U_2,P_2)$ be two solutions to $D_u U_{xx} + J(U)=0$ with $0<P_1(0)<P_2(0)$, $U_1(0)=U_2(0)=0$. Let $x_j>0$ be smallest so that $P_j(x_j)=0$, $j=1,2$. We need to show that $x_1<x_2$.


The strict decay of $P_2$ implies that there is smallest $x_*>0$ so that $P_2(x_*)=P_1(0)$. See Figure~\ref{f:large}(b). From the ordering of trajectories noted in the proof of Theorem~\ref{t:sestr} we have that $P_1(y_1)=P_2(y_2)$, implies $U_1(y_1)<U_2(y_2)$ for $y_j\in(0,x_j]$. Now $\H>0$ and $\H'\leq 0$ gives $0> -\H(U_2(y_2)) \geq -\H(U_1(y_1))$. Hence, $\partial_x P_j =-\H(U_j)$, $j=1,2$ implies $P_1(x) \leq P_2(x_*+x)$ for all $x\in[0,x_1]$ and so $x_2>x_1$. 
%
\end{proof}

\subsection{Separated boundary conditions and unbounded $D$}~

\begin{figure}
\centering
\scalebox{1}{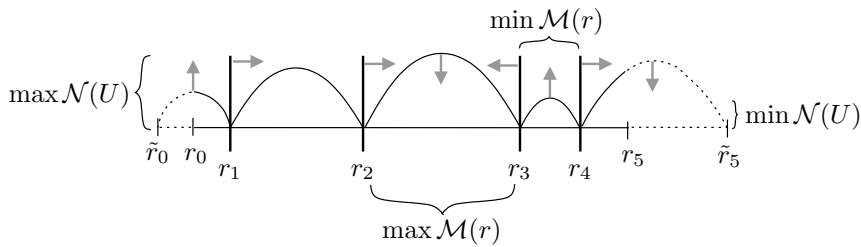} 
\caption{Illustration of the Lyapunov functions and interface dynamics for bounded $D$ with separated boundary conditions and $N=5$. Dashed lines mark the extensions beyond $D$. The extension on the left generates local minima of $\calN(U)$ and $\calM(r)$, and that on the right local maxima.}\label{f:Lyap-sep}
\end{figure}

In order to make the boundary segments $U_0$ and $U_N$ comparable with the interpulse segments, we define $\tilde r_0:= r_1 + \Delta_-(\partial_x U(r_1-))$ and  $\tilde r_{N+1}:= r_N + \Delta_+(\partial_x U(r_N+))$. We further adjust Definition~\ref{d:mima} and the definitions of $\calN(U)$, $\calM(r)$ to non-periodic domains as follows: (1) replace $r_0$ with $\tilde r_0$ and $r_{N+1}$ by $\tilde r_{N+1}$, (2) set $U^*_{-1}=U^*_0$, $d_{-1} =d_0$ and $U^*_{N+2}=U^*_{N+1}$, $d_{N+2} =d_{N+1}$. 

\begin{theorem}\label{t:sep}
Assume Hypothesis~\ref{h:DN} and linear separated boundary conditions. 
Theorem~\ref{t:sestr} holds with respect to the adjusted $\calN(U)$, $\calM(r)$.
%
%
Moreover, if $(\tilde r_0, \tilde r_{N+1})$ is bounded, then Corollary \ref{c:sestr} holds for the Lyapunov functionals defined on $[\tilde r_0, \tilde r_{N+1}]$ with the adjusted $\calN(U)$, $\calM(r)$. 
\end{theorem}

\begin{proof}
With the adjustments the proofs are the same as those of Theorem~\ref{t:sestr} and Corollary~\ref{c:sestr}, respectively.
%
\end{proof}


\medskip
\begin{theorem}
Assume unbounded $D$, $r\in\calD_N$ and that $U$ converges to constant states in unbounded directions. Let $j\in\{0,N+1\}$. If $|r_j|=\infty$ then $\sgn(\frac{\rmd}{\rmd\tau}r_j)$ is constant, and nonzero if $N>1$. If in addition $\H$ has a unique saddle in $\{U>0\}$ and $N>1$ then $\sgn(\frac{\rmd}{\rmd\tau}r_j) = \sgn (r_j)$.
%
\end{theorem}

Note that the models in \S\ref{s:apex} have a unique saddle if $\mu>0$. 

\begin{proof}
We consider only $r_0=-\infty$ as $r_{N+1}=\infty$ follows by symmetry. The convergence to an equilibrium implies that this must be a saddle point (since $n=1$) and $U_0(x):=U(x)$ for $x\leq r_0$ equals its unstable manifold up to the point $(U_0(r_1), \partial_x U_0(r_1-))$. Therefore, $\partial_x U_0(r_1-)$ is constant in time. The stable mani\-fold is the reflection of the unstable one about the $U$-axis, and forms a separatrix in the planar phase space. Hence, for all $\tau$ either $\partial_x U_0(r_1+) \neq -\partial_x U_0(r_1-)$, or $\partial_x U_0(r_1+)\equiv -\partial_x U_0(r_1-)$. The latter implies that $U_1$ lies in the stable manifold of the saddle, which requires $N=1$ and $\frac{\rmd}{\rmd\tau}r_1=0$. 

Concerning the former, following the arguments of the proof of Theorem~\ref{t:sestr} (trajectory ordering and the definition of $\calD_N$) shows that $\sgn(\frac{\rmd}{\rmd\tau}r_1)$ is constant in time. See Figure~\ref{f:large}(a).

For $N>1$, the point $(U(r_1),\partial_xU(r_1+))$ cannot lie in the stable manifold of the saddle point, and thus $\partial_x U_0(r_1-) \neq -\partial_x U_0(r_1+)$, which implies $\frac{\rmd}{\rmd\tau} r_1\neq0$ by Hypothesis~\ref{h:DN}(3).

If the saddle is unique in $\{U>0\}$ then $U(r_1)=0$ together with $U>0$ for $x< r_1$ implies  that the saddle has $U$-coordinate $U_*>0$. Uniqueness of the saddle and that its stable and unstable manifolds form separatrices imply the following: trajectories with $U\geq 0$ which intersect $\{U=0\}$ intersect the $U$-axis in the interval $[0,U_*)$. See Figure~\ref{f:large}(a). As in the proof of Theorem~\ref{t:sestr} it follows that $\frac{\rmd}{\rmd\tau} r_1< 0$ if $N>1$. 
\end{proof}

\end{document}

%% file: sestr-sym-bif.pstex_t
\begin{picture}(0,0)%
\epsfig{file=sestr-sym-bif.pstex}%
\end{picture}%
\setlength{\unitlength}{3947sp}%
\begingroup\makeatletter\ifx\SetFigFont\undefined%
\gdef\SetFigFont#1#2#3#4#5{%
  \reset@font\fontsize{#1}{#2pt}%
  \fontfamily{#3}\fontseries{#4}\fontshape{#5}%
  \selectfont}%
\fi\endgroup%
\begin{picture}(9333,5750)(2026,-7836)
\put(2851,-4711){\makebox(0,0)[rb]{\smash{{\SetFigFont{20}{24.0}{\rmdefault}{\mddefault}{\updefault}{\color[rgb]{0,0,0}$\|\cdot\|_2$}%
}}}}
\put(4726,-5461){\makebox(0,0)[lb]{\smash{{\SetFigFont{20}{24.0}{\rmdefault}{\mddefault}{\updefault}{\color[rgb]{0,0,0}unimodal}%
}}}}
\put(5701,-6061){\makebox(0,0)[lb]{\smash{{\SetFigFont{20}{24.0}{\rmdefault}{\mddefault}{\updefault}{\color[rgb]{0,0,0}bimodal}%
}}}}
\put(7876,-5311){\makebox(0,0)[lb]{\smash{{\SetFigFont{20}{24.0}{\rmdefault}{\mddefault}{\updefault}{\color[rgb]{0,0,0}$P^\rms_\pf$}%
}}}}
\put(8026,-3886){\makebox(0,0)[lb]{\smash{{\SetFigFont{20}{24.0}{\rmdefault}{\mddefault}{\updefault}{\color[rgb]{0,0,0}$C_\rms>0$}%
}}}}
\put(7973,-6040){\makebox(0,0)[lb]{\smash{{\SetFigFont{20}{24.0}{\rmdefault}{\mddefault}{\updefault}{\color[rgb]{0,0,0}$C_\rms<0$}%
}}}}
\put(9826,-6811){\makebox(0,0)[rb]{\smash{{\SetFigFont{20}{24.0}{\rmdefault}{\mddefault}{\updefault}{\color[rgb]{0,0,0}$P^\rms_\sn$}%
}}}}
\put(6901,-7561){\makebox(0,0)[lb]{\smash{{\SetFigFont{20}{24.0}{\rmdefault}{\mddefault}{\updefault}{\color[rgb]{0,0,0}$P^\rms$}%
}}}}
\end{picture}%

%% file: sestr-cusp.pstex_t
\begin{picture}(0,0)%
\epsfig{file=sestr-cusp.pstex}%
\end{picture}%
\setlength{\unitlength}{3947sp}%
\begingroup\makeatletter\ifx\SetFigFont\undefined%
\gdef\SetFigFont#1#2#3#4#5{%
  \reset@font\fontsize{#1}{#2pt}%
  \fontfamily{#3}\fontseries{#4}\fontshape{#5}%
  \selectfont}%
\fi\endgroup%
\begin{picture}(9333,5750)(2026,-7836)
\put(6751,-7636){\makebox(0,0)[lb]{\smash{{\SetFigFont{20}{24.0}{\rmdefault}{\mddefault}{\updefault}{\color[rgb]{0,0,0}$P^\rma$}%
}}}}
\put(2176,-4561){\makebox(0,0)[lb]{\smash{{\SetFigFont{20}{24.0}{\rmdefault}{\mddefault}{\updefault}{\color[rgb]{0,0,0}$P^\rms$}%
}}}}
\put(9470,-2834){\makebox(0,0)[rb]{\smash{{\SetFigFont{20}{24.0}{\rmdefault}{\mddefault}{\updefault}{\color[rgb]{0,0,0}$(P^\rms_\#,P^\rma_\#)$}%
}}}}
\end{picture}%

%% file: linear.pstex_t
\begin{picture}(0,0)%
\epsfig{file=linear.pstex}%
\end{picture}%
\setlength{\unitlength}{3947sp}%
\begingroup\makeatletter\ifx\SetFigFont\undefined%
\gdef\SetFigFont#1#2#3#4#5{%
  \reset@font\fontsize{#1}{#2pt}%
  \fontfamily{#3}\fontseries{#4}\fontshape{#5}%
  \selectfont}%
\fi\endgroup%
\begin{picture}(9333,5772)(2026,-7858)
\put(6751,-5011){\makebox(0,0)[lb]{\smash{{\SetFigFont{20}{24.0}{\rmdefault}{\mddefault}{\updefault}{\color[rgb]{0,0,0}$P^\rms_\sn$}%
}}}}
\put(2701,-4711){\makebox(0,0)[rb]{\smash{{\SetFigFont{20}{24.0}{\rmdefault}{\mddefault}{\updefault}{\color[rgb]{0,0,0}$\|\cdot\|_2$}%
}}}}
\put(4201,-5311){\makebox(0,0)[lb]{\smash{{\SetFigFont{20}{24.0}{\rmdefault}{\mddefault}{\updefault}{\color[rgb]{0,0,0}unimodal}%
}}}}
\put(4951,-6211){\makebox(0,0)[lb]{\smash{{\SetFigFont{20}{24.0}{\rmdefault}{\mddefault}{\updefault}{\color[rgb]{0,0,0}bimodal}%
}}}}
\put(7426,-3586){\makebox(0,0)[lb]{\smash{{\SetFigFont{20}{24.0}{\rmdefault}{\mddefault}{\updefault}{\color[rgb]{0,0,0}$C_\rms>0$}%
}}}}
\put(8176,-5986){\makebox(0,0)[lb]{\smash{{\SetFigFont{20}{24.0}{\rmdefault}{\mddefault}{\updefault}{\color[rgb]{0,0,0}$C_\rms<0$}%
}}}}
\put(6901,-7786){\makebox(0,0)[lb]{\smash{{\SetFigFont{20}{24.0}{\rmdefault}{\mddefault}{\updefault}{\color[rgb]{0,0,0}$P^\rms$}%
}}}}
\end{picture}%

%% file: sestr-arcs-per.pstex_t
\begin{picture}(0,0)%
\epsfig{file=sestr-arcs-per.pstex}%
\end{picture}%
\setlength{\unitlength}{3947sp}%
\begingroup\makeatletter\ifx\SetFigFont\undefined%
\gdef\SetFigFont#1#2#3#4#5{%
  \reset@font\fontsize{#1}{#2pt}%
  \fontfamily{#3}\fontseries{#4}\fontshape{#5}%
  \selectfont}%
\fi\endgroup%
\begin{picture}(3430,1482)(4245,-3423)
\put(7175,-2622){\makebox(0,0)[lb]{\smash{{\SetFigFont{10}{12.0}{\rmdefault}{\mddefault}{\updefault}{\color[rgb]{0,0,0}$\min \calN(U)$}%
}}}}
\put(5967,-3378){\makebox(0,0)[b]{\smash{{\SetFigFont{10}{12.0}{\rmdefault}{\mddefault}{\updefault}{\color[rgb]{0,0,0}$\max \calM(r)$}%
}}}}
\put(5401,-2986){\makebox(0,0)[lb]{\smash{{\SetFigFont{10}{12.0}{\rmdefault}{\mddefault}{\updefault}{\color[rgb]{0,0,0}$r_2$}%
}}}}
\put(4576,-2986){\makebox(0,0)[lb]{\smash{{\SetFigFont{10}{12.0}{\rmdefault}{\mddefault}{\updefault}{\color[rgb]{0,0,0}$r_1$}%
}}}}
\put(4366,-2887){\makebox(0,0)[lb]{\smash{{\SetFigFont{10}{12.0}{\rmdefault}{\mddefault}{\updefault}{\color[rgb]{0,0,0}$r_0$}%
}}}}
\put(6376,-2986){\makebox(0,0)[lb]{\smash{{\SetFigFont{10}{12.0}{\rmdefault}{\mddefault}{\updefault}{\color[rgb]{0,0,0}$r_3$}%
}}}}
\put(6751,-2986){\makebox(0,0)[lb]{\smash{{\SetFigFont{10}{12.0}{\rmdefault}{\mddefault}{\updefault}{\color[rgb]{0,0,0}$r_4$}%
}}}}
\put(6619,-2061){\makebox(0,0)[b]{\smash{{\SetFigFont{10}{12.0}{\rmdefault}{\mddefault}{\updefault}{\color[rgb]{0,0,0}$\min \calM(r)$}%
}}}}
\put(6999,-2890){\makebox(0,0)[lb]{\smash{{\SetFigFont{10}{12.0}{\rmdefault}{\mddefault}{\updefault}{\color[rgb]{0,0,0}$r_5=r_0$}%
}}}}
\put(4245,-2529){\makebox(0,0)[rb]{\smash{{\SetFigFont{10}{12.0}{\rmdefault}{\mddefault}{\updefault}{\color[rgb]{0,0,0}$\max \calN(U)$}%
}}}}
\end{picture}%

%% file: phase1.pstex_t
\begin{picture}(0,0)%
\epsfig{file=phase1.pstex}%
\end{picture}%
\setlength{\unitlength}{3947sp}%
\begingroup\makeatletter\ifx\SetFigFont\undefined%
\gdef\SetFigFont#1#2#3#4#5{%
  \reset@font\fontsize{#1}{#2pt}%
  \fontfamily{#3}\fontseries{#4}\fontshape{#5}%
  \selectfont}%
\fi\endgroup%
\begin{picture}(4317,2823)(826,-4117)
\put(4651,-3961){\makebox(0,0)[lb]{\smash{{\SetFigFont{20}{24.0}{\rmdefault}{\mddefault}{\updefault}{\color[rgb]{0,0,0}$U$}%
}}}}
\put(826,-2761){\makebox(0,0)[lb]{\smash{{\SetFigFont{20}{24.0}{\rmdefault}{\mddefault}{\updefault}{\color[rgb]{0,0,0}$P^*$}%
}}}}
\put(826,-1561){\makebox(0,0)[lb]{\smash{{\SetFigFont{20}{24.0}{\rmdefault}{\mddefault}{\updefault}{\color[rgb]{0,0,0}$P$}%
}}}}
\put(3901,-4036){\makebox(0,0)[lb]{\smash{{\SetFigFont{20}{24.0}{\rmdefault}{\mddefault}{\updefault}{\color[rgb]{0,0,0}$U_*$}%
}}}}
\put(2552,-3979){\makebox(0,0)[lb]{\smash{{\SetFigFont{20}{24.0}{\rmdefault}{\mddefault}{\updefault}{\color[rgb]{0,0,0}$U^*$}%
}}}}
\end{picture}%

%% file: phase2.pstex_t
\begin{picture}(0,0)%
\epsfig{file=phase2.pstex}%
\end{picture}%
\setlength{\unitlength}{3947sp}%
\begingroup\makeatletter\ifx\SetFigFont\undefined%
\gdef\SetFigFont#1#2#3#4#5{%
  \reset@font\fontsize{#1}{#2pt}%
  \fontfamily{#3}\fontseries{#4}\fontshape{#5}%
  \selectfont}%
\fi\endgroup%
\begin{picture}(5035,2754)(144,-4048)
\put(3826,-2761){\makebox(0,0)[lb]{\smash{{\SetFigFont{20}{24.0}{\rmdefault}{\mddefault}{\updefault}{\color[rgb]{0,0,0}$P_2(x_*)$}%
}}}}
\put(1221,-2355){\makebox(0,0)[rb]{\smash{{\SetFigFont{20}{24.0}{\rmdefault}{\mddefault}{\updefault}{\color[rgb]{0,0,0}$P_2(0)$}%
}}}}
\put(1222,-2742){\makebox(0,0)[rb]{\smash{{\SetFigFont{20}{24.0}{\rmdefault}{\mddefault}{\updefault}{\color[rgb]{0,0,0}$P_1(0)$}%
}}}}
\put(4651,-3961){\makebox(0,0)[lb]{\smash{{\SetFigFont{20}{24.0}{\rmdefault}{\mddefault}{\updefault}{\color[rgb]{0,0,0}$U$}%
}}}}
\put(826,-1561){\makebox(0,0)[lb]{\smash{{\SetFigFont{20}{24.0}{\rmdefault}{\mddefault}{\updefault}{\color[rgb]{0,0,0}$P$}%
}}}}
\end{picture}%

%% file: sestr-arcs.pstex_t
\begin{picture}(0,0)%
\epsfig{file=sestr-arcs.pstex}%
\end{picture}%
\setlength{\unitlength}{3947sp}%
\begingroup\makeatletter\ifx\SetFigFont\undefined%
\gdef\SetFigFont#1#2#3#4#5{%
  \reset@font\fontsize{#1}{#2pt}%
  \fontfamily{#3}\fontseries{#4}\fontshape{#5}%
  \selectfont}%
\fi\endgroup%
\begin{picture}(3819,1482)(4033,-3423)
\put(7852,-2656){\makebox(0,0)[lb]{\smash{{\SetFigFont{10}{12.0}{\rmdefault}{\mddefault}{\updefault}{\color[rgb]{0,0,0}$\min \calN(U)$}%
}}}}
\put(5967,-3378){\makebox(0,0)[b]{\smash{{\SetFigFont{10}{12.0}{\rmdefault}{\mddefault}{\updefault}{\color[rgb]{0,0,0}$\max \calM(r)$}%
}}}}
\put(5401,-2986){\makebox(0,0)[lb]{\smash{{\SetFigFont{10}{12.0}{\rmdefault}{\mddefault}{\updefault}{\color[rgb]{0,0,0}$r_2$}%
}}}}
\put(4576,-2986){\makebox(0,0)[lb]{\smash{{\SetFigFont{10}{12.0}{\rmdefault}{\mddefault}{\updefault}{\color[rgb]{0,0,0}$r_1$}%
}}}}
\put(4366,-2887){\makebox(0,0)[lb]{\smash{{\SetFigFont{10}{12.0}{\rmdefault}{\mddefault}{\updefault}{\color[rgb]{0,0,0}$r_0$}%
}}}}
\put(4127,-2886){\makebox(0,0)[lb]{\smash{{\SetFigFont{10}{12.0}{\rmdefault}{\mddefault}{\updefault}{\color[rgb]{0,0,0}$\tilde r_0$}%
}}}}
\put(6376,-2986){\makebox(0,0)[lb]{\smash{{\SetFigFont{10}{12.0}{\rmdefault}{\mddefault}{\updefault}{\color[rgb]{0,0,0}$r_3$}%
}}}}
\put(6751,-2986){\makebox(0,0)[lb]{\smash{{\SetFigFont{10}{12.0}{\rmdefault}{\mddefault}{\updefault}{\color[rgb]{0,0,0}$r_4$}%
}}}}
\put(7679,-2910){\makebox(0,0)[lb]{\smash{{\SetFigFont{10}{12.0}{\rmdefault}{\mddefault}{\updefault}{\color[rgb]{0,0,0}$\tilde r_5$}%
}}}}
\put(6619,-2061){\makebox(0,0)[b]{\smash{{\SetFigFont{10}{12.0}{\rmdefault}{\mddefault}{\updefault}{\color[rgb]{0,0,0}$\min \calM(r)$}%
}}}}
\put(7084,-2903){\makebox(0,0)[lb]{\smash{{\SetFigFont{10}{12.0}{\rmdefault}{\mddefault}{\updefault}{\color[rgb]{0,0,0}$r_5$}%
}}}}
\put(4033,-2528){\makebox(0,0)[rb]{\smash{{\SetFigFont{10}{12.0}{\rmdefault}{\mddefault}{\updefault}{\color[rgb]{0,0,0}$\max \calN(U)$}%
}}}}
\end{picture}%